\documentclass[conference,letter]{IEEEtran}
\usepackage[T1]{fontenc}
\usepackage[utf8x]{inputenc}
\usepackage[english]{babel}
\usepackage{cite}
\usepackage{amsmath}
\usepackage{times}

\usepackage{amsthm}
\usepackage[labelformat=simple]{subcaption}

\usepackage[ruled,vlined,linesnumbered]{algorithm2e}
\usepackage{graphicx}
\usepackage{tikz}
\usepackage{booktabs}
\usepackage{diagbox}
\usepackage{color}

\DeclareRobustCommand*{\IEEEauthorrefmark}[1]{%
  \raisebox{0pt}[0pt][0pt]{\textsuperscript{\footnotesize\ensuremath{#1}}}}

\theoremstyle{plain}
\newtheorem{definition}{Definition}
\newtheorem{lemma}{Lemma}
\newtheorem{theorem}{Theorem}

\newtheorem{proposition}{Proposition}

\def \lo{\textit{low}-criticality }
\def \hi{\textit{high}-criticality }

\def \ull{U_{LO}^{LO} }
\def \ulh{U_{LO}^{HI}}
\def \uhl{U_{HI}^{LO} }
\def \uhh{U_{HI}^{HI} }

\begin{document}

\def \papertitle{EDF-VD Scheduling of Mixed-Criticality Systems with Degraded Quality Guarantees}

\title{\papertitle}

\author{\small
\authorblockN{
Di Liu\IEEEauthorrefmark{1},
Jelena Spasic\IEEEauthorrefmark{1}
Gang Chen\IEEEauthorrefmark{2},
Nan Guan\IEEEauthorrefmark{3},
Songran Liu\IEEEauthorrefmark{2},
Todor Stefanov\IEEEauthorrefmark{1},
Wang Yi\IEEEauthorrefmark{2,4}}
\authorblockA{\IEEEauthorrefmark{1} Leiden University, The Netherlands}
\authorblockA{\IEEEauthorrefmark{2} Northeastern University, China}
\authorblockA{\IEEEauthorrefmark{3} Hong Kong Polytechnic University, Hong Kong}
\authorblockA{\IEEEauthorrefmark{4} Uppsala University, Sweden}
}

\maketitle

\begin{abstract}
This paper studies real-time scheduling of mixed-criticality systems
where low-criticality tasks are still guaranteed some service in the
high-criticality mode, with reduced execution budgets.
First, we present a utilization-based schedulability test for such
systems under EDF-VD scheduling. Second, we quantify the suboptimality of
EDF-VD (with our test condition) in terms of speedup factors.
In general, the speedup factor is a function with respect to the ratio between the amount of resource 
required by different types of tasks in different criticality modes, and reaches 4/3 in the worst case.
Furthermore, we show that the proposed utilization-based schedulability test and speedup factor results 
apply to the elastic mixed-criticality model as well.
Experiments show effectiveness of our proposed method and
confirm the theoretical suboptimality results.

\end{abstract}

\section{Introduction}
\label{section:introduce}
An important trend in real-time embedded systems is to integrate applications with different criticality levels into a shared platform in order to reduce resource cost and energy consumption. To ensure the correctness of a \emph{mixed-criticality} (MC) system, highly critical tasks are subject to certification by Certification Authorities under extremely rigorous and pessimistic assumptions \cite{Baruah2010}. This generally causes large worst-case
execution time (WCET) estimation for high-criticality tasks.
On the other hand, the system designer needs to consider the timing requirement
of the entire system, but under less conservative assumptions.
The challenge in scheduling MC systems is to simultaneously
guarantee the timing correctness of (1) only high-criticality tasks
under very pessimistic assumptions, and (2) all tasks, including low-critical
ones, under less pessimistic assumptions.

The scheduling problem of MC systems has been intensively studied in recent years (see Section \ref{section:related} for a brief review).
In most of previous works, the timing correctness of high-criticality tasks
are guaranteed in the worst case scenario at the expense of low-criticality tasks.
More specifically, when any high criticality task executes for longer than its low-criticality WCET (and thus the system enters the
high-criticality mode), all low-criticality tasks will be completely
discarded and the resource are all dedicated for meeting the timing constraints of
high-criticality tasks \cite{Baruah2012,Baruah2011Response,Ekberg2013article,Easwaran2013}.  
However, such an approach seriously disturbs the service of low-criticality tasks. This
is not acceptable in many practical problems, especially for control systems
where the performance of controllers mainly depends on
the execution frequency of control tasks \cite{Su2013}.

To overcome this problem, 
Burns and Baruah in \cite{burns2013towards} introduced an MC task model where low-criticality tasks reduce their execution budgets such that low-criticality tasks are guaranteed to be scheduled in high-criticality mode with regular execution frequency (i.e., the same period) but degraded quality\footnote{In \cite{Liu1991}\cite{Liu1994}, the output quality of a task is related to its execution time. The longer a task executes, the better quality results it produces.}.
Since the idea of reducing execution budgets to keep tasks running is conceptually similar to the \textit{imprecise computation model} \cite{Liu1991}\cite{Liu1994},
in this paper we call an MC system with possibly reduced execution budgets of low criticality tasks an \textit{imprecise mixed criticality} (IMC) system.

In \cite{burns2013towards}, the authors consider preemptive fixed-priority scheduling for the IMC system model
and extend the adaptive mixed criticality (AMC) \cite{Baruah2011Response} approach to provide a schedulability test for the IMC model.
In this paper, we study the EDF-VD scheduling of IMC systems.
EDF-VD is designed for the classical MC system model,
in which EDF is enhanced by deadline adjustment mechanisms to compromise the
resource requirement on different criticality levels. EDF-VD has shown strong competence by both theoretical and empirical evaluations \cite{Baruah2012,Ekberg2013article,Easwaran2013}. For example, \cite{Baruah2012} proves that EDF-VD is a speedup-optimal MC scheduling and in \cite{Ekberg2013article,Easwaran2013} experimental evaluations show that EDF-VD outperforms other MC scheduling algorithms
in terms of acceptance ratio.
The main technical contributions of this paper include
\begin{itemize}
 \item  We propose a sufficient test for the IMC model under EDF-VD, - see Theorem \ref{theorem:overall} in Section \ref{section:sufficientTest};

 \item For the IMC model under EDF-VD, we derive a speedup factor function with respect to the utilization ratios of high criticality tasks and low criticality tasks
- see Theorem \ref{theorem:speedup} in Section \ref{section:speed}.
The derived speedup factor function enables us to quantify the suboptimality of EDF-VD and evaluate the impact of the utilization ratios on the speedup factor.
We also compute the maximum value $4/3$ of the speedup factor function,  which is equal to the speedup factor bound for the classical MC model \cite{Baruah2012}.

 \item With extensive experiments, we show that for the IMC model, by using our proposed sufficient test,
in most cases EDF-VD outperforms AMC \cite{burns2013towards} in terms of the number of schedulable task sets. Moreover,
the experimental results validate the observations we obtained for speedup factor.

\end{itemize}

Moreover, the schedulability test and speedup factor results of this paper also apply to the \textit{elastic mixed-criticality} (EMC) model proposed in \cite{Su2013}, where the periods of 
low-criticality tasks are scaled up in high-criticality mode, see in Section \ref{section:emc}.

The remainder of this paper is organized as follows:
Section \ref{section:related} discusses the related work.
Section \ref{section:background} gives the preliminaries and describes the IMC task model and its execution semantics.
Section \ref{section:sufficientTest} presents our sufficient test for the IMC model and Section
\ref{section:speed} derives the speedup factor function for the IMC under EDF-VD.
Section \ref{section:emc} extends the proposed sufficient test to the EMC model.
Finally, Section \ref{section:evaluation} shows our experimental results and Section \ref{section:conclusion} concludes this paper.

\section{Related Work}
\label{section:related}
Burns and Davis in \cite{burns2015mixed} give a comprehensive review of work on real-time scheduling for MC systems.
 Many of these literatures, e.g., \cite{Baruah2012}\cite{Ekberg2013article}\cite{Easwaran2013}, consider the classical MC model in which all low criticality tasks are discarded if the system switches to the high
criticality mode.
In \cite{burns2013towards}, Burns and Baruah discuss three approaches to keep some low criticality tasks running in \hi mode.
The first approach is to change the priority of low criticality tasks.
However, for fixed-priority scheduling, deprioritizing low criticality tasks cannot guarantee the execution of the low criticality tasks with a short deadline after the mode switches.  \cite{burns2013towards}.
Similarly, for EDF, lowering priority of low criticality tasks leads to a degraded service \cite{Huang2014}.
In this paper, we consider the IMC model which improves the schedulability of low criticality tasks in \hi mode by reducing their execution time.
The IMC model can guarantee the regular service of a system by trading off the quality of the produced results.
For some applications given in \cite{Liu1991}\cite{Liu1994}\cite{Ravindran2014}, such trade-off is preferred.
The second approach in \cite{burns2013towards} is to extend the periods of low criticality tasks when the system mode changes to \hi mode such that the low criticality tasks execute less frequently to ensure their schedulability.
Su \textit{et al.} \cite{Su2013}\cite{Su2014} and Jan \textit{et al.} \cite{jan2013maximizing} both consider this model.
However, some applications might prefer an on-time result with a degraded quality rather than a delayed result with a perfect quality.
Some example applications can be seen in \cite{Chung1990}\cite{Liu1991}\cite{Liu1994}.
Then, the approach of extending periods is less useful for this kind of applications.
The last approach proposed in \cite{burns2013towards} is to reduce the execution budget of low criticality tasks when the system mode switches, i.e., the use of the IMC model studied in this paper.
In \cite{burns2013towards}, the authors extend the AMC \cite{Baruah2011Response} approach to test the schedulability of an IMC task set under fixed-priority scheduling.
However, the schedulability problem for an IMC task set under EDF-VD \cite{Baruah2012}, has not yet been addressed.
Therefore, in this paper, we study the schedulability of the IMC task model under EDF-VD and propose a sufficient test for it.

\section{Preliminaries}
\label{section:background}
This section first introduces the IMC task model and its execution semantics. 
Then, we give a brief explanation for EDF-VD scheduling \cite{Baruah2012} and an example to illustrate the execution semantics of the IMC model under EDF-VD scheduling.
\subsection{Imprecise Mixed-Criticality Task Model}
\label{background:IMC}

We use the \textit{implicit-deadline sporadic} task model given in \cite{burns2013towards}
where a task set $\gamma$ includes $n$ tasks which are scheduled on a uniprocessor.
Without loss of generality, all tasks in $\gamma$ are assumed to start at time 0.
Each task $\tau_i$ in $\gamma$ generates an infinite sequence of jobs $\{J_i^1,J_i^2 ...\}$ and is characterized by $\tau_i = \{T_i,D_i,L_i,\mathcal{C}_i\}$:
\begin{itemize}
 \setlength\itemsep{0.1em}
 \item  $T_i$ is the period or the minimal separation interval between two consecutive jobs;
 \item  $D_i$ denotes the relative task deadline, where $D_i = T_i$;
  \item  $L_i \in \{LO, HI\}$ denotes the criticality (\textit{low or high}) of a task. In this paper, like in many previous research works \cite{Su2013}\cite{Huang2014}\cite{Baruah2012}\cite{Ekberg2013article}\cite{Easwaran2013},
  we consider a duel-criticality MC model. Then, we split tasks into two task sets, $\gamma_{LO} = \{\tau_i|L_i = LO\}$ and $\gamma_{HI} = \{\tau_i|L_i = HI\}$;
  \item  $\mathcal{C}_i = \{C_i^{LO},C_i^{HI}\}$ is a list of WCETs,
where $C_i^{LO}$ and $C_i^{HI}$ represent the WCET in \lo mode and the WCET in \hi mode, respectively.
For a \hi task, it has $C_i^{LO} \le C_i^{HI}$,
whereas \textit{$C_i^{LO} \ge C_i^{HI}$ for a \lo task, i.e., \lo task $\tau_i$ has a reduced WCET in \hi mode}.
\end{itemize}
Then each job $J_i$ is characterized by $J_i = \{a_i,d_i,L_i,\mathcal{C}_i\}$, where $a_i$ is the absolute release time and $d_i$ is the absolute deadline. 
Note that if \lo task $\tau_i$ has $C_i^{HI} = 0$, it will be immediately discarded at the time of the switch to \hi mode.
In this case, the IMC model behaves like the classical MC model.

The utilization of a task is used to denote the ratio between its WCET and its period.
We define the following utilizations for an IMC task set $\gamma$:
\begin{itemize}
 \setlength\itemsep{0.1em}
 \item For every task $\tau_i$, it has $u_i^{LO} = \frac{C_i^{LO}}{T_i}$, $u_i^{HI} = \frac{C_i^{HI}}{T_i}$;
 \item For all \lo tasks, we have total utilizations
\[\tiny \ull = \sum_{\forall \tau_i \in \gamma_{LO}}u_i^{LO}, \quad \ulh = \sum_{\forall \tau_i \in \gamma_{LO}}u_i^{HI}\]
 \item For all \hi tasks, we have total utilizations
\[\tiny \uhl = \sum_{\forall \tau_i \in \gamma_{HI}}u_i^{LO}, \quad \uhh = \sum_{\forall \tau_i \in \gamma_{HI}}u_i^{HI}\]
\item For an IMC task set, we have
 \[U^{LO} = \ull+\uhl, \quad U^{HI} = \ulh+\uhh\]
\end{itemize}

\subsection{Execution Semantics of the IMC Model}
\label{background:schematic}

The execution semantics of the IMC model are similar to those of the classical MC model.
The \textbf{major difference} occurs after a system switches to \hi mode.
\textit{Instead of discarding all \lo tasks, as it is done in the classical MC model, the IMC model tries to schedule \lo tasks with their reduced execution times $C_i^{HI}$. }
The execution semantics of the IMC model are summarized as follows:
\begin{itemize}
 \setlength\itemsep{0.1em}
\item The system starts in \lo mode, and remains in this mode as long as no \hi job overruns its \lo WCET $C_i^{LO}$.
If any job of a \lo task tries to execute beyond its $C_i^{LO}$, the system will suspend it and launch a new job at the next period;
\item If any job of \hi task executes for its $C_i^{LO}$ time units without signaling completion,
the system immediately switches to \hi mode;
\item As the system switches to \hi mode, \textit{if jobs of \lo tasks have completed execution for more than their $C_i^{HI}$ but less than their $C_i^{LO}$, the jobs will be suspended till the tasks release new jobs for the next period.
However, if jobs of \lo tasks have not completed their $C_i^{HI}$ ($\le C_i^{LO}$) by the switch time instant,
the jobs will complete the left execution to $C_i^{HI}$ after the switch time instant and before their deadlines. Hereafter, all jobs are scheduled using $C_i^{HI}$.}
For \hi tasks, if their jobs have not completed their $C_i^{LO}$ ($\le C_i^{HI}$) by the switch time instant,
all jobs will continue to be scheduled to complete $C_i^{HI}$. After that, all jobs are scheduled using $C_i^{HI}$.
\end{itemize}
Santy \textit{et al.} \cite{Santy2012} have shown that the system can switch back from the \hi mode to the \lo mode when there is an idle period and no \hi job awaits for execution.
For the IMC model, we can use the same
scenario to trigger the switch-back. In this paper, we focus on the switch from \lo mode to \hi mode.

\subsection{EDF-VD Scheduling}
The challenge to schedule MC tasks with EDF scheduling algorithm \cite{liu1973scheduling} is to deal with the overrun of \hi tasks when the system switches from \lo mode to \hi  mode. 
 Baruah \textit{et al.} in \cite{Baruah2012} proposed to artificially tighten deadlines of jobs of \hi tasks in \lo mode 
such that the system can preserve execution budgets for the \hi tasks across mode switches. This approach is called \textit{EDF with virtual deadlines} (EDF-VD).

\subsection{An Illustrative Example}

\begin{table}
\centering
 \begin{tabular}{|l|l|l|l|l|l|}
\hline
   Task & $L$ & $C_i^{LO}$ & $C_i^{HI}$ & $T_i$ & $\hat{D_i}$\\ \hline
  $\tau_1$ & LO & 4 & 2 & 9 & \\ \hline
  $\tau_2$ & HI & 4 & 7 & 10 & 7  \\ \hline
 \end{tabular}
\caption{Illustrative example}
\label{table:example_Imc}
\end{table}

\begin{figure}
 \centering
\includegraphics[width=0.8\columnwidth]{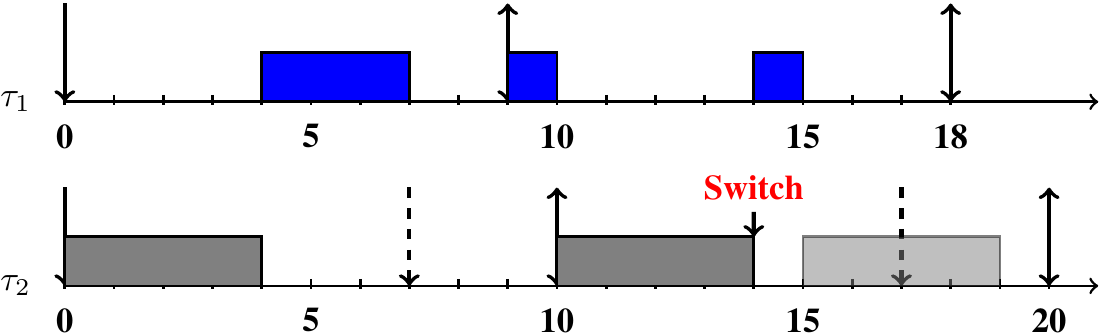}
\caption{Scheduling of Example \ref{table:example_Imc}}
\label{fig:imc_example}
\end{figure}

Here, we give a simple example to illustrate the execution semantics of the IMC model under EDF-VD.
Table \ref{table:example_Imc} gives two tasks, one \lo task $\tau_1$ and one \hi task $\tau_2$,
where $\hat{D_i}$ is the virtual deadline. Figure \ref{fig:imc_example} depicts the scheduling
of the given IMC task set, where we assume that the mode switch occurs in the second period of $\tau_2$. When the system switches to \hi mode, $\tau_2$ will be scheduled by its original deadline $10$ instead of its virtual deadline $7$. Hence, $\tau_1$ preempts $\tau_2$ at the switch time instant. Since in \hi mode $\tau_1$ only has execution budget of 2 , i.e., $C_1^{HI}$, $\tau_1$ executes one unit and suspends. Then, $\tau_2$ completes its left execution $4$ ($C_2^{HI}-C_2^{LO}$) before its deadline.

\section{Schedulability Analysis}
\label{section:sufficientTest}
In \cite{burns2013towards}, an AMC-based schedulability test for the IMC model under fixed priority scheduling has been proposed.
However, to date, a schedulability test for the IMC model under EDF-VD
has not been addressed yet.
Therefore, inspired by the work in \cite{Baruah2012} for the classical MC model, we propose a sufficient test for the IMC model under EDF-VD.

\subsection{Low Criticality Mode}
\label{lowmode}
We first ensure the schedulability of tasks when they are in \lo  mode.
As the task model is in \lo  mode, the tasks can be considered as traditional real-time tasks scheduled by EDF with virtual deadlines (VD). The following theorem is given in \cite{Baruah2012} for tasks scheduled in \lo  mode.
\begin{theorem}[Theorem 1 from \cite{Baruah2012}]
\label{theorem:lowcrit}
 The following condition is sufficient for ensuring that EDF-VD successfully schedules all tasks in \lo  mode:
 {\small
 \begin{equation}
  \label{equation:edvVD_low_bound}
 x \ge \frac{\uhl}{1 - \ull}
 \end{equation}
 }
where $x\in(0,1)$ is used to uniformly modify the relative deadline of \hi tasks.
\end{theorem}
Since the IMC model behaves as the classical MC model in \lo  mode, Theorem \ref{theorem:lowcrit} holds for the IMC model as well.

\subsection{High Criticality Mode}
For \hi mode, the classical MC model discards all \lo jobs after the switch to \hi mode.
In contrast, the IMC model keeps \lo jobs running but with degraded quality, i.e., a shorter execution time.
So the schedulability condition in \cite{Baruah2012} does not work for the IMC model in the \hi  mode.
Thus, we need a new test for the IMC model in \hi mode.

\begin{figure}
 \centering
 \includegraphics[width=0.8\columnwidth]{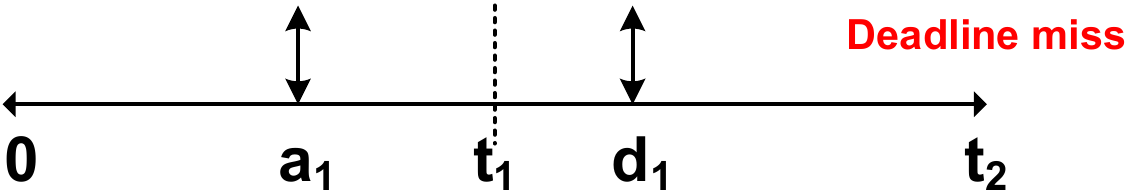}
 \caption{The defined time instants}
 \label{figure:definition}
\end{figure}

To derive the sufficient test in \hi mode, suppose that there is a time interval $[0, t_2]$,
where a first deadline miss occurs at $t_2$ and $t_1$ denotes the time instant of the switch to \hi mode in the time interval, where $t_1 < t_2$.
Assume that $\mathcal{J}$ 
is the minimal set of jobs generated from task set $\gamma$ which leads to the first deadline miss at $t_2$.
The minimality of $\mathcal{J}$ means that removing any job in $\mathcal{J}$ guarantees the schedulability of the rest of $\mathcal{J}$.
Here, we introduce some notations for our later interpretation.
Let variable $\eta_i$ denote the cumulative execution time of task $\tau_i$ in the interval $[0,t_2]$.
$J_1$ denotes a special \hi job which has switch time instant $t_1$ within its period $(a_1, d_1)$, i.e, $a_1 < t_1 < d_1$.
Furthermore, $J_1$ is the job with the earliest release time amongst all \hi jobs in $\mathcal{J}$ which execute in $[t_1,t_2)$. Figure \ref{figure:definition} visualizes the defined time instants.
Moreover, we define a special type of job for \lo tasks which is useful for our later proofs.
\begin{definition}
 A job $J_i$ from \lo task $\tau_i$ is a \textsf{carry-over job}, if its absolute release time $a_i$ is before and its absolute deadline $d_i$ is 
 after the switch time instant, i.e., $a_i < t_1 < d_i$.
\end{definition}

With the notations introduced above, we have the following propositions,
\begin{proposition}[Fact 1 from \cite{Baruah2012}]
 \label{proposition:t1_t2}
 All jobs in $\mathcal{J}$ that execute in $[t_1,t_2)$ have deadline $\le t_2$.
\end{proposition}
It is easy to observe that only jobs which have deadlines $\le t_2$ are possible to cause a deadline miss at $t_2$.
If a job has its deadline $>t_2$ and is still in set $\mathcal{J}$, it will contradict the minimality of $\mathcal{J}$.

\begin{proposition}
 \label{proposition:case_1_2}
The switch time instant $t_1$ has
 \begin{equation}
\label{equation:case_1_2}
 t_1 < (a_1+x(t_2-a_1))
 \end{equation}
\end{proposition}
\begin{proof}
 Let us consider a time instant $(a_1+ x(d_1 - a_1))$ which is the virtual deadline of job $J_1$.
Since $J_1$ executes in time interval $[t_1,t_2)$, its virtual deadline $(a_1 + x(d_1 - a_1))$ must be greater than the switch time instant $t_1$.
Otherwise, it should have completed its \lo execution before $t_1$,
and this contradicts that it executes in $[t_1,t_2)$.
Thus, it holds that
\[\begin{split}
& t_1 < (a_1 + x(d_1 - a_1)) \\
\Rightarrow & t_1 < (a_1+x(t_2-a_1)) \quad (\text{since }d_1 \le t_2)
  \end{split}
\]
\end{proof}

\begin{proposition}
\label{proposition:carry}
 If a \textsf{carry-over} job $J_i$ has its cumulative execution equal to $(d_i-a_i)u_i^{LO}$ and  $u_i^{LO} >u_i^{HI}$, its deadline $d_i$ is $\le (a_1 + x(t_2 -a_1))$.
\end{proposition}
\begin{proof}
For a \textsf{carry-over} job $J_i$, if it has its cumulative execution equal to $(d_i-a_i)u_i^{LO}$ and $u_i^{LO} > u_i^{HI}$, it should complete its $C_i^{LO}$ execution before $t_1$.
Otherwise, if job $J_i$ has executed time units $C_i\in [C_i^{HI},C_i^{LO})$ at time instant $t_1$, it will be suspended and will not execute after $t_1$.

Now, we will show that when job $J_i$ completes its $C_i^{LO}$ execution, its deadline is $d_i \le (a_1 + x(t_2 -a_1))$.
We prove this by contradiction.
 First, we suppose that $J_i$ has its deadline $d_i > (a_1+x(t_2 -a_1))$  and release time $a_i$.
As shown above, job $J_i$ completes its $C_i^{LO}$ execution before $t_1$.
 Let us assume a time instant $t^*$ as the latest time instant at which this \textsf{carry-over job} $J_i$ starts to execute before $t_1$.
 This means that at this time instant all jobs in $\mathcal{J}$ with deadline $\le (a_1+x(t_2 -a_1))$ have finished their executions.
This indicates that these jobs will not have any execution within interval $[t^*,t_2]$.
Therefore, jobs in $\mathcal{J}$ with release time at or after time instant $t^*$ can form a smaller job set which causes a deadline miss at $t_2$.
Then, it contradicts the minimality of $\mathcal{J}$.
Thus, \textsf{carry-over} job $J_i$ with its cumulative execution time equal to $(d_i - a_i)u_i^{LO}$ and $u_i^{LO}>u_i^{HI}$ has its deadline $d_i \le (a_1 +x(t_2-a_1))$.
\end{proof}

With the propositions and notations given above, we derive an upper bound of the cumulative execution time $\eta_i$ of
\lo task $\tau_i$.
\begin{lemma}
 \label{lemma:low_crit}
 For any \lo task $\tau_i$, it has
 \begin{equation}
  \label{equation:lowbound}
  \eta_i \le (a_1 + x(t_2 - a_1))u_i^{LO} + (1 - x)(t_2 - a_1)u_i^{HI}
 \end{equation}
\end{lemma}
\begin{proof}
If $u_i^{LO}=u_i^{HI}$, it is trivial to see that Lemma \ref{lemma:low_crit} holds.
Below we focus on the case when $u_i^{LO}>u_i^{HI}$.
If a system switches to \hi mode at $t_1$, then we know that \lo tasks are scheduled using $C_i^{LO}$ before $t_1$ and
using $C_i^{HI}$ after $t_1$.
To prove this lemma, we need to consider two cases, where $\tau_i$ releases a job within interval $(a_1,t_2]$ or it does not.
We prove the two cases separately.

\textbf{Case A} (task $\tau_i$ releases a job within interval $(a_1,t_2]$):
There are two sub-cases to be considered.

\begin{itemize}
 \item \noindent{\bf Sub-case 1 (No carry-over job):} The deadline of a job of \lo task $\tau_i$ coincides with switch time instant $t_1$.
The cumulative execution time of \lo task $\tau_i$ within time interval $[0,t_2]$ can be bounded as follows,
\[\eta_i \le (t_1 -0) \cdot u_i^{LO} + (t_2 - t_1) \cdot u_i^{HI} \]
Since $t_1 < (a_1 + x(t_2 - a_1))$ according to Proposition \ref{proposition:case_1_2} and for \lo task $\tau_i$ it has $u_i^{LO} > u_i^{HI}$, then
\[\small
 \begin{split}
  &\eta_i < \big(a_1 + x(t_2 - a_1)\big)u_i^{LO} + \big(t_2 - \big(a_1 + x(t_2 - a_1)\big)\big)u_i^{HI} \\
  \Leftrightarrow & \eta_i < (a_1 + x(t_2 - a_1))u_i^{LO} + (1 - x)(t_2 - a_1)u_i^{HI}
  \end{split}
\]

\item\noindent{\bf Sub-case 2 (with carry-over job):} In this case, before the \textsf{carry-over} job, jobs of $\tau_i$ are scheduled with its $C_i^{LO}$.
After the \textsf{carry-over} job, jobs of $\tau_i$ are scheduled with its $C_i^{HI}$.
It is trivial to observe that for a \textsf{carry-over} job its maximum cumulative execution time can be obtained when it completes its $C_i^{LO}$ within its period $[a_i,d_i]$,
i.e., $(d_i-a_i)u_i^{LO}$.
Considering the maximum cumulative execution for the \textsf{carry-over job}, we then have for \lo task $\tau_i$,
\[
\begin{split}
 &\eta_i  \le (a_i-0) u_i^{LO} + (d_i - a_i) u_i^{LO} + (t_2 - d_i) u_i^{HI}  \\
 \Leftrightarrow &\eta_i  \le d_i u_i^{LO} + (t_2 - d_i) u_i^{HI} \\
\end{split}
\]
Proposition \ref{proposition:carry} shows as $J_i$ has its cumulative execution equal to $(d_i-a_i)\cdot u_i^{LO}$, it has $d_i \le (a_1 +x(t_2-a_1))$.
Given $u_i^{LO}> u_i^{HI}$ for \lo task,  we have
\[
\small
\begin{split}
 &\eta_i  \le d_i u_i^{LO} + (t_2 - d_i) u_i^{HI} \\
\Rightarrow & \eta_i \le \big(a_1 + x(t_2 - a_1)\big)u_i^{LO} + \big(t_2 - \big(a_1 + x(t_2 - a_1)\big)\big)u_i^{HI} \\
 \Leftrightarrow &\eta_i  \le (a_1 + x(t_2 - a_1))u_i^{LO} + (1 - x)(t_2 - a_1)u_i^{HI}
\end{split}
\]

\end{itemize}
\textbf{Case B} (task $\tau_i$ does not release a job within interval $(a_1,t_2]$): In this case,
let $J_i$ denote the last release job of task $\tau_i$ before $a_1$ and $a_i$ and $d_i$ are its absolute release time and absolute deadline, respectively.
If $d_i \le t_1$, we have
\[\eta_i =(a_i-0)u_i^{LO} + (d_i-a_i)\cdot u_i^{LO} = d_iu_i^{LO}\]
If $d_i > t_1$, $J_i$ is a \textsf{carry-over job}. As we discussed above, the maximum cumulative execution time of \textsf{carry-over job} $J_i$ is $(d_i-a_i)u_i^{LO}$,
so we have
\[
\label{eq:imc_1}
 \eta_i \le (a_i-0)u_i^{LO} + (d_i-a_i)\cdot u_i^{LO}\Leftrightarrow \eta_i \le d_iu_i^{LO}
\]
Similarly, according to Proposition \ref{proposition:carry}, we obtain,
\[
\small
 \begin{split}
  & \eta_i \le d_i \cdot u_i^{LO} \le (a_1 +x(t_2-a_1))u_i^{LO} \\
\Rightarrow & \eta_i < (a_1 +x(t_2-a_1))u_i^{LO} + \big(t_2 - \big(a_1 + x(t_2 - a_1)\big)\big)u_i^{HI}  \\
 \Leftrightarrow & \eta_i < (a_1 + x(t_2 - a_1))u_i^{LO} + (1 - x)(t_2 - a_1)u_i^{HI} \\
 \end{split}
\]

\end{proof}

 Lemma \ref{lemma:low_crit} gives the upper bound of the cumulative execution time of a \lo task in \hi mode.
In order to derive the sufficient test for the IMC model in \hi mode, we need to upper bound the cumulative execution time of
\hi tasks.

\begin{proposition}[Fact 3 from \cite{Baruah2012}]
\label{proposition:hi_crit}
 For any \hi task $\tau_i$, it holds that
 \begin{equation}
  \eta_i \le \frac{a_1}{x}u_i^{LO} + (t_2 -a_1)u_i^{HI}
 \end{equation}
\end{proposition}

Proposition \ref{proposition:hi_crit} is used to bound the cumulative execution of the \hi tasks.
Since in the IMC model the \hi tasks are scheduled as in the classical MC model, Proposition \ref{proposition:hi_crit} holds for the IMC model as well.
With Lemma \ref{lemma:low_crit} and Proposition \ref{proposition:hi_crit}, we can derive the sufficient test for the IMC model in \hi mode.
\begin{theorem}
 \label{theorem:high_crit}
 The following condition is sufficient for ensuring that EDF-VD successfully schedules all tasks in \hi mode:
\begin{equation}
 \label{equation:edfvd_upper}
 x\ull + (1 - x) \ulh + \uhh \le 1
\end{equation}
\end{theorem}

\begin{proof}
Let $N$ denote the cumulative execution time of all tasks in $\gamma = \gamma_{LO} \cup \gamma_{HI}$ over interval $[0,t_2]$. We have
\[N = \sum_{\forall \tau_i \in \gamma_{LO}}\eta_i + \sum_{\forall \tau_i \in \gamma_{HI}}\eta_i \]
By using Lemma \ref{lemma:low_crit} and Proposition \ref{proposition:hi_crit}, $N$ is bounded as follows
{\small
  \begin{equation}
 \label{equation:inequation_1}
 \begin{split}
 &N \le \sum_{\forall \tau_i \in \gamma_{LO}}\bigg(\big(a_1 + x(t_2 - a_1)\big)u_i^{LO} + (1 - x)(t_2 - a_1)u_i^{HI}\bigg) \\
 & + \sum_{\forall \tau_i \in \gamma_{HI}}\bigg(\frac{a_1}{x}u_i^{LO} + (t_2 -a_1)u_i^{HI}\bigg) \\
 \Leftrightarrow & N \le  (a_1 + x(t_2 - a_1))\ull + (1 - x)(t_2 - a_1)U^{HI}_{LO} \\
 &+ \frac{a_1}{x}\uhl + (t_2 - a_1)\uhh \\
\Leftrightarrow &N \le a_1(\ull + \frac{\uhl}{x}) + x(t_2 - a_1)\ull  \\
   &  + (1-x)(t_2 - a_1)\ulh + (t_2 - a_1)\uhh
 \end{split}
\end{equation}}

Since the tasks must be schedulable in \lo  mode, the condition given in Theorem \ref{theorem:lowcrit} holds and
we have $ 1 \ge (\ull + \frac{\uhl}{x})$. Hence,
{\small
 \begin{equation}
 \label{equation:inequation_3}
  \begin{split}
 N  \le & a_1 + x(t_2 - a_1)\ull\\
 &   + (1-x)(t_2 - a_1)\ulh + (t_2 - a_1)\uhh
  \end{split}
 \end{equation}}
Since time instant $t_2$ is the first deadline miss, it means that there is no idle time instant within interval $[0,t_2]$.
Note that if there is an idle instant, jobs from set $\mathcal{J}$ which have release time at or after the latest idle instant
can form a smaller job set causing deadline miss at $t_2$ which contradicts the minimality of $\mathcal{J}$.
Then, we obtain
\[\small
 \begin{split}
 & N = \bigg (\sum_{\forall \tau_i \in \gamma_{LO}}\eta_i + \sum_{\forall \tau_i \in \gamma_{HI}}\eta_i \bigg) > t_2 \\
 \Rightarrow & a_1 + x(t_2 - a_1)\ull + (1-x)(t_2 - a_1)\ulh + (t_2 - a_1)\uhh \\ & > t_2 \\
 \Leftrightarrow & x(t_2 - a_1)\ull + (1-x)(t_2 - a_1)\ulh + (t_2 - a_1)\uhh \\ & > t_2 - a_1 \\
 \Leftrightarrow & x\ull + (1-x)\ulh + \uhh > 1
\end{split}
\]
By taking the contrapositive, we derive the sufficient test for the IMC model when it is in \hi mode:
{\small
 \[x\ull + (1-x)\ulh + \uhh \le 1\]}
\end{proof}

Note that if $\ulh = 0$, i.e., no \lo tasks are scheduled after the system switches to \hi mode,
our Theorem \ref{theorem:high_crit} is the same as the sufficient test (Theorem 2 in \cite{Baruah2012}) for the classical MC model in \hi mode.
Hence, our Theorem \ref{theorem:high_crit} actually is a generalized schedulability condition for (I)MC tasks under EDF-VD.

By combining Theorem \ref{theorem:lowcrit} (see Section \ref{lowmode}) and our Theorem \ref{theorem:high_crit}, we prove the following theorem,
\begin{theorem}
\label{theorem:overall}
 Given an IMC task set, if
 {\small
  \begin{equation}
   \label{equation:total}
   \uhh + \ull \le 1
  \end{equation}
then the IMC task set is schedulable by EDF; otherwise, if
 \begin{equation}
  \label{equation:normalcase}
  \frac{\uhl}{1-\ull} \le \frac{1 - (\uhh + \ulh)}{\ull - \ulh}
 \end{equation}}
where
 \begin{equation}
  \label{equation:worstcase}
  \uhh + \ulh < 1 \text{ and } \ull < 1 \text{ and } \ull > \ulh
 \end{equation}
 then this IMC task set can be scheduled by EDF-VD with a deadline scaling factor $x$ arbitrarily chosen in
 the following range
 \[
x \in \left[\frac{\uhl}{1-\ull},~~\frac{1 - (\uhh + \ulh)}{\ull - \ulh} \right]
 \]
\end{theorem}

 \begin{proof}
Total utilization $U \le 1$ is the exact test for EDF on a uniprocessor system. If the condition in (\ref{equation:total}) is met, the given task set is \textit{worst-case reservation} \cite{Baruah2012} schedulable under EDF, i.e., the task set can be scheduled by EDF without deadline scaling for \hi tasks and execution budget reduction for \lo tasks.
Now, we prove the second condition given by (\ref{equation:normalcase}).
 From Theorem \ref{theorem:lowcrit}, we have,
 \[
  \label{equation:x_low}
  x \ge \frac{\uhl}{1 - \ull}
 \]
 From Theorem \ref{theorem:high_crit}, we have
 \[
\begin{split}
& x\ull + (1 - x) \ulh + \uhh \le 1 \\
\Leftrightarrow & x \le \frac{1 - (\uhh + \ulh)}{\ull - \ulh}
\end{split}
\label{equation:proof}
 \]
Therefore, if $\frac{\uhl}{1-\ull} \le \frac{1 - (\uhh + \ulh)}{\ull - \ulh}$, the schedulability conditions of both Theorem \ref{theorem:lowcrit} and \ref{theorem:high_crit} are satisfied.
Thus, the IMC tasks are schedulable under EDF-VD.
 \end{proof}
 

\section{Speedup Factor}
\label{section:speed}

The speedup factor bound is a useful metric to compare the worst-case performance of different MC scheduling algorithms.
The speedup factor bound for the classical MC model under EDF-VD \cite{Baruah2012} has been shown to be $4/3$.
The following is the definition of the speedup factor for an MC scheduling algorithm.

\begin{definition}[from \cite{Baruah2012}]
An algorithm $\mathcal{A}$ has a speedup factor $f \geq 1$, if any task system
that is schedulable on a unit-speed processor by using a hypothetical optimal clairvoyant scheduling algorithm\footnote{A `clairvoyant' scheduling algorithm knows
all run-time information, e.g., when the mode switch will occur, prior to run-time.}, can be successfully scheduled on a speed-$f$ processor by algorithm $\mathcal{A}$.
\end{definition}
%

For notational simplicity, we define
\begin{align*}
\uhh = c, ~~~~&\uhl=\alpha\times c \\
\ull = b, ~~~~&\ulh=\lambda \times b
\end{align*}
where $\alpha \in(0,1]$ and $\lambda \in [0,1]$. $\alpha$ denotes the utilization ratio between $\uhl$ and $\uhh$, while $\lambda$ denotes the utilization ratio between  $\ulh$ and $\ull$.

First, let us analyze the speedup factor of two corner cases. When $\alpha=1$, i.e., $\uhl = \uhh$, this means that there is no mode-switch. Therefore, the task set is scheduled by the traditional EDF, i.e., the task set is schedulable  if $\ull +\uhl \le 1$. Since EDF is the optimal scheduling algorithm on a uniprocessor system,
the speedup factor is 1. When $\lambda=1$, i.e., $\ull=\ulh$, if the task set is schedulable in \hi mode, it must hold $\uhh + \ull \le 1$ by Theorem \ref{theorem:high_crit}.
Then it is scheduled by the traditional EDF and thus the speedup factor is 1 as well.

In this paper, instead of generating a single speedup factor bound,
we derive a speedup factor function with respect to $(\alpha,\lambda)$.
This speedup factor function enables us to quantify the suboptimality of EDF-VD for the IMC model in terms of speedup factor (by our proposed sufficient test) and evaluate the impact of the utilization ratio on the schedulability of an IMC task set under EDF-VD.

First, we strive to find a minimum speed $s$ ($\le$1) for a clairvoyant optimal MC scheduling algorithm such that any implicit-deadline IMC task set
which is schedulable by the clairvoyant optimal MC scheduling algorithm on a speed-$s$ processor can satisfy the schedulability test given in Theorem \ref{theorem:overall}, i.e., schedulable under EDF-VD on a unit-speed processor.

\begin{lemma}
\label{lemma:speedup}
 Given $b,c\in[0,1]$, $\alpha \in (0,1)$, $\lambda \in [0,1)$, and
\begin{equation}
 \max\{b+\alpha c,\lambda b +c\} \le S(\alpha,\lambda)
\end{equation}
where
\[S(\alpha,\lambda) = \frac{(1-\alpha \lambda)((2-\alpha \lambda-\alpha) +(\lambda-1)\sqrt{4\alpha-3\alpha^2})}{2(1-\alpha)(\alpha \lambda-\alpha \lambda^2-\alpha+1)}\]
then it guarantees
\begin{equation}
\label{speedup:condition}
 \frac{\alpha c}{1-b} \le \frac{1-(c+\lambda b)}{b - \lambda b}
\end{equation}
\end{lemma}
\begin{proof}
Suppose that $s \ge \max\{b+\alpha c,\lambda b +c\}$. We strive to find a minimal value $s$ to guarantee that (\ref{speedup:condition}) in Theorem \ref{theorem:overall} is always
satisfied. Based on this,
we construct an optimization problem as follows,
 \begin{align}
 \label{speedup:nonconvex_problem_1}
  \text{minimize}\quad & s \\
\text{subject to}\quad & b+\alpha c \le s \label{plane1}\\
& \lambda b +c \le s  \label{plane2}\\
& \lambda b^2 + (\alpha \lambda - \alpha + 1)bc - (\lambda+1)b - c + 1 \le 0 \label{convex:1}\\
& 0 \le b \le 1,\quad 0 \le c \le 1
 \end{align}
where $\alpha$ and $\lambda$ are constant and $s,b,c$ are variables.
If we can prove that $S(\alpha,\lambda)$ is the optimal solution of the optimization problem (\ref{speedup:nonconvex_problem_1}),
then Lemma \ref{lemma:speedup} is proved.

Below, we prove that $S(\alpha,\lambda)$ is
the optimal solution of the optimization problem (\ref{speedup:nonconvex_problem_1})\footnote{
This optimization problem is a non-convex problem and thus we cannot use general convex optimization techniques such as the
Karush-Kuhn-Tucker (KKT) approach \cite{kuhn1951} to solve it. 
}.
%
\begin{figure}[t]
\begin{subfigure}[b]{0.31\columnwidth}
 \centering
\includegraphics[width=1\columnwidth]{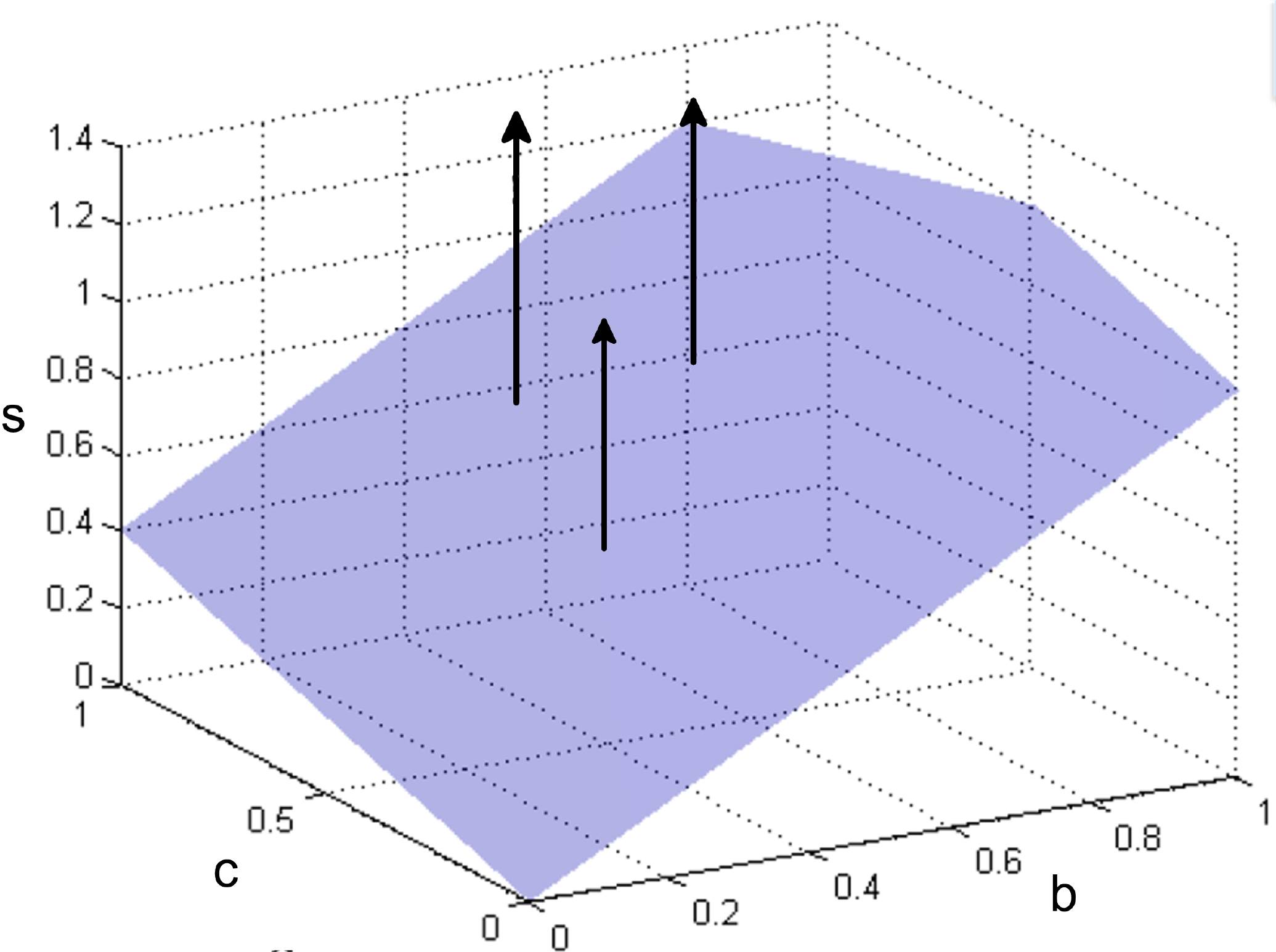}
\caption{plane 1}
\label{fig:plane1}
\end{subfigure}
\begin{subfigure}[b]{0.31\columnwidth}
 \centering
\includegraphics[width=1\columnwidth]{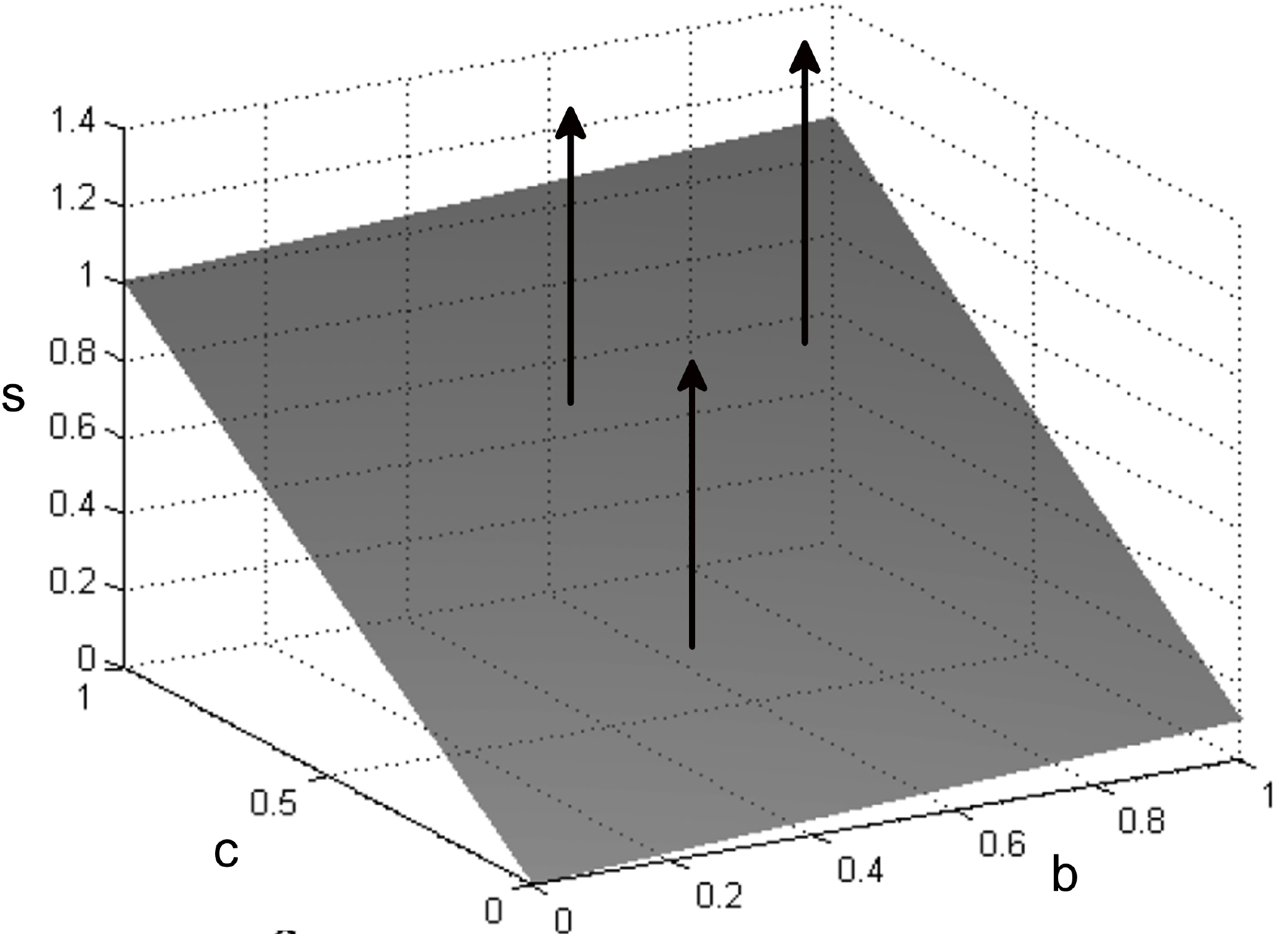}
\caption{plane 2}
\label{fig:plane2}
\end{subfigure}
\begin{subfigure}[b]{0.31\columnwidth}
 \centering
\includegraphics[width=1\columnwidth]{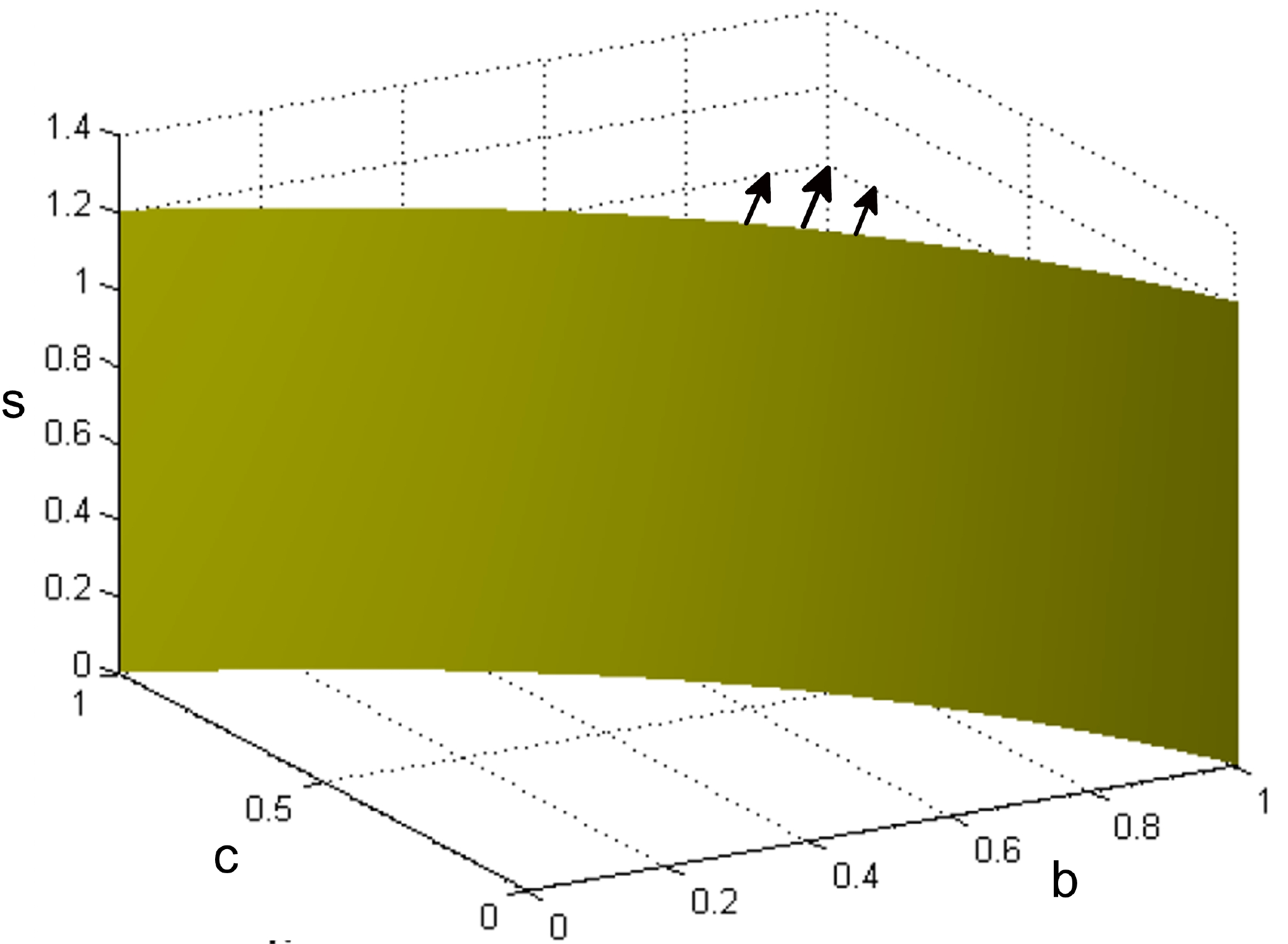}
\caption{vertical surface}
\label{fig:vertical}
\end{subfigure}

\begin{subfigure}[b]{1\columnwidth}
 \centering
\includegraphics[width=0.75\columnwidth]{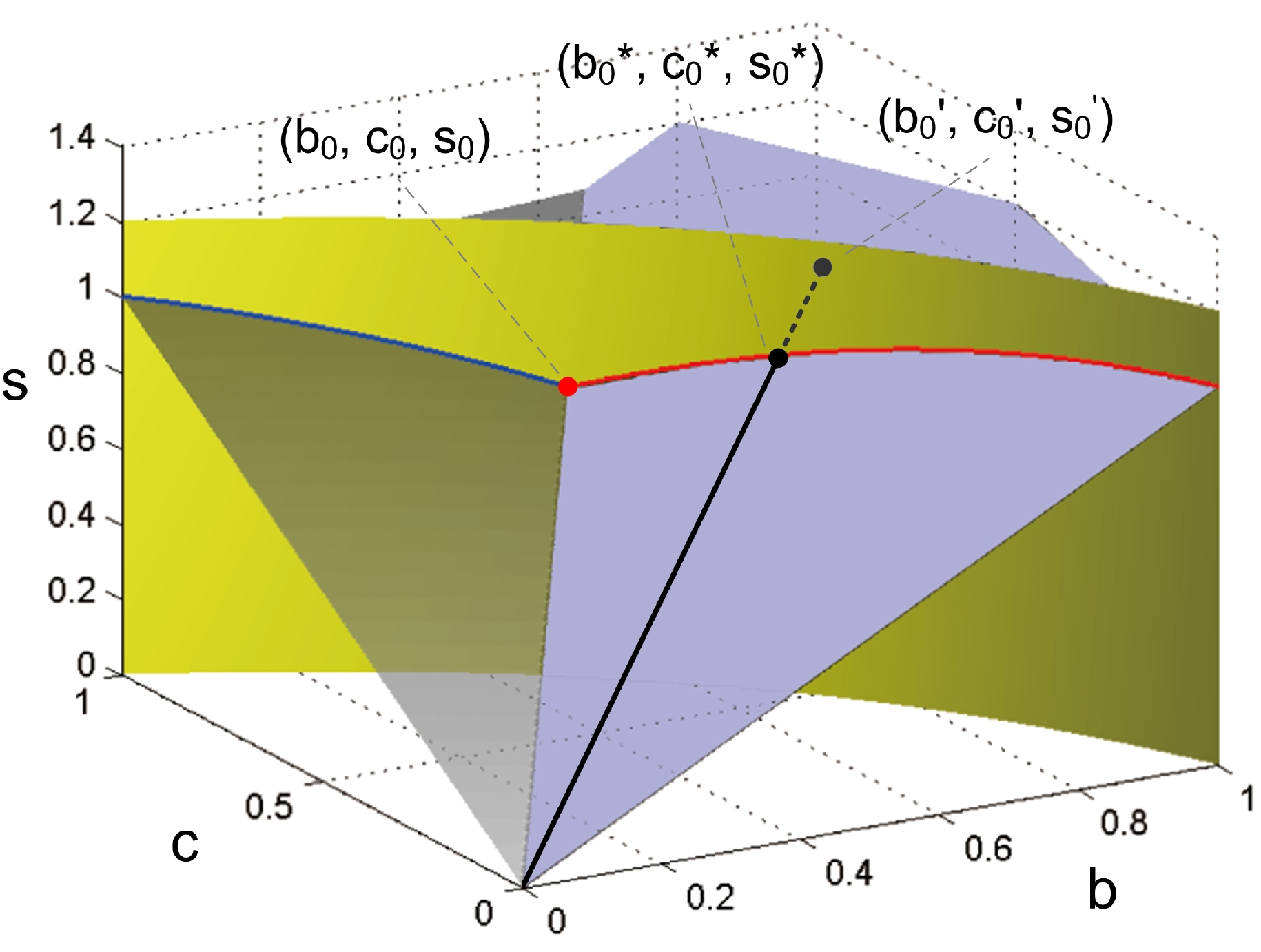}
\caption{Feasible solution space}
\label{fig:feasible}
\end{subfigure}
\caption{3D space of optimization problem (\ref{speedup:nonconvex_problem_1})}
\end{figure}
Each of constraints in the optimization problem (\ref{speedup:nonconvex_problem_1}) defines a feasible space in the three-dimension space.
In Figure \ref{fig:plane1}, the space above the plane is a feasible space satisfying constraint (\ref{plane1}), where the plane corresponds to $b+\alpha c = s$.
For constraint (\ref{plane2}), $\lambda b+c = s$ draws a plane and the feasible space is above the plane shown in Figure \ref{fig:plane2}.
Similarly, when constraint (\ref{convex:1}) makes its right-hand-side equal to the left-hand-side, we draw a vertical curved surface seen in Figure \ref{fig:vertical} and the space outside the vertical surface is the feasible space\footnote{As the arrows direct.}.
As a result, the feasible solutions subject to these three constraints must be above both planes and
outside the vertical curved surfaces. Below we will prove the minimum value of $s$ in the feasible solution space must be on the vertical surface and one plane.

First assume that we have a point $(b_0',c_0',s_0')$ which satisfies all constraints but
is not on the vertical surface. If we connect the origin $(0,0,0)$ and $(b_0',c_0',s_0')$, this line must have an intersection point $(b_0^*,c_0^*,s_0^*)$ with the
vertical surface. It is easy to observe that $s_0^* < s_0'$ - see in Figure \ref{fig:feasible}.
This means that any point which is not on the vertical surface can find a point with smaller value of  $s$ on the vertical surface which satisfies all constraints. Therefore,
the point with the minimal $s$ must be on the vertical surface.
Similarly, the minimal $s$ must be on one of the two planes. Otherwise, if it is not on
any plane, we always can find a projected point on one plane which has a smaller value of $s$.

We have shown above that to obtain the minimal value of $s$ the point must be on the vertical surface and one plane.
Then, the two planes have an intersection line and this line intersects with the vertical surface at a point denoted by $(b_0,c_0,s_0)$.
By taking constraints (\ref{plane1})(\ref{plane2}) and (\ref{convex:1}), we formulate a piece-wise function of $s$ with
respect to $b$ as follows.
\begin{equation}
\label{eq:piecewise}
s(b) =
\left\{
\begin{array}{ll}
\frac{(\alpha \lambda^2 - \alpha \lambda)b^2 + b - 1}{(\alpha \lambda - \alpha + 1)b - 1}  &  0< b \leq b_0 \\
\frac{(1-\alpha)b^2 + (\alpha \lambda +\alpha-1) b - \alpha}{(\alpha \lambda - \alpha + 1)b - 1}  & b_0<  b \leq 1
\end{array} \right.
\end{equation}
This function covers all points which are on the vertical surface and one plane and at same time satisfy all constraints.
By doing some calculus, we know that Eq. (\ref{eq:piecewise}) is monotonically decreasing in $(0,b_0]$ and monotonically increasing in $(b_0,1]$.
Therefore, the minimal value of Eq. (\ref{eq:piecewise}) can be obtained at $(b_0,c_0,s_0)$. 
The complete proof is given by Lemma \ref{lemma:minimal_value} in Appendix I. 
It means that we can obtain the optimal solution of optimization problem (\ref{speedup:nonconvex_problem_1}) by solving the following equation system.

\begin{equation}\label{e:equation-system}
\begin{cases}
~b_0+xc_0 = s_0 \\
~\lambda b_0+c_0 = s_0 \\
~\lambda b_0^2 \!+ \! (\!\alpha\lambda \!-\! \alpha \!+\! 1\!)b_0c_0 \!-\!(\lambda \! + \!1\!)b_0 \!-\! c_0 \!+ \!1  = 0
\end{cases}
\end{equation}
By joining the first two equations we have $c_0 = \frac{1-\lambda}{1-\alpha} \times b_0$, and applying it to the last equation in (\ref{e:equation-system}) gives
\[
(-\alpha \lambda^2 + \alpha \lambda - \alpha +1) b_0^2 + (\alpha\lambda + \alpha - 2)b_0 + (1 - \alpha) = 0
\]

	By the well-known Quadratic Formula we get the two roots
	of the above quadratic equation.
	\begin{align}
       \small
	b_0^1 & = \frac{(2 - \alpha \lambda - \alpha) + (1-\lambda)\sqrt{-3\alpha^2 + 4\alpha }}{2(-\alpha \lambda^2 + \alpha \lambda - \alpha +1)} \label{e:root-one}\\
	b_0^2 & = \frac{(2 - \alpha \lambda - \alpha) - (1-\lambda)\sqrt{-3\alpha^2 + 4\alpha }}{2(-\alpha \lambda^2 + \alpha \lambda - \alpha +1)} \label{e:root-two}
	\end{align}

	We can prove that $b_0^2$ is larger than 1 and thus should be dropped (since we
	require $0\le b \le 1$), while $b_0^1$ is in the range of $[0,1]$. 
	The detailed proof is given by Lemma \ref{lemma:out_range} in Appendix I.
	As a result, we obtain the optimal solution $(b_0^1,\frac{1-\alpha}{1-\lambda}b_0^1,\frac{1-\alpha \lambda}{1-\lambda}b_0^1)$ for Eq. (\ref{e:equation-system}). Thus, we have

\[\small
\begin{split}
 S(\alpha,\lambda) & = \frac{1-\alpha\lambda}{1-\lambda}b_0^1 \\
& = \frac{(1-\alpha \lambda)((2-\alpha \lambda-\alpha) +(\lambda-1)\sqrt{4\alpha-3\alpha^2})}{2(1-\alpha)(\alpha \lambda-\alpha \lambda^2-\alpha+1)}
\end{split}
\]
Therefore, Lemma \ref{lemma:speedup} is proved.
\end{proof}

Lemma \ref{lemma:speedup} shows that any IMC task set which is schedulable by an optimal clairvoyant MC scheduling
algorithm on
a speed-$S(\alpha,\lambda)$ is also schedulable on a unit-speed processor by EDF-VD.
As a direct result, we have the following theorem,
\begin{theorem}
 \label{theorem:speedup}
 The speedup factor of EDF-VD with IMC task sets is upper bounded by
\[f = \frac{2(1-\alpha)(\alpha \lambda-\alpha \lambda^2-\alpha+1)}{(1-\alpha \lambda)((2-\alpha \lambda-\alpha) +(\lambda-1)\sqrt{4\alpha-3\alpha^2})}\]
\end{theorem}
The speedup factor is shown to be a function with respect to $\alpha$ and $\lambda$. Figure \ref{speedup_fig} plots the 3D image of this function
 and Table \ref{speedup_table} lists some of the values with different $\alpha$ and $\lambda$. By doing some calculus, we obtain the maximum value $1.333$, i.e., $4/3$, of the speedup factor function when $\lambda=0$ and $\alpha = \frac{1}{3}$,  which is highlighted in Figure \ref{speedup_fig} and Table \ref{speedup_table}.
We see that the speedup factor bound is achieved when the task set is a classical MC task set.
\begin{figure}
 \centering
 \includegraphics[width=0.6\columnwidth]{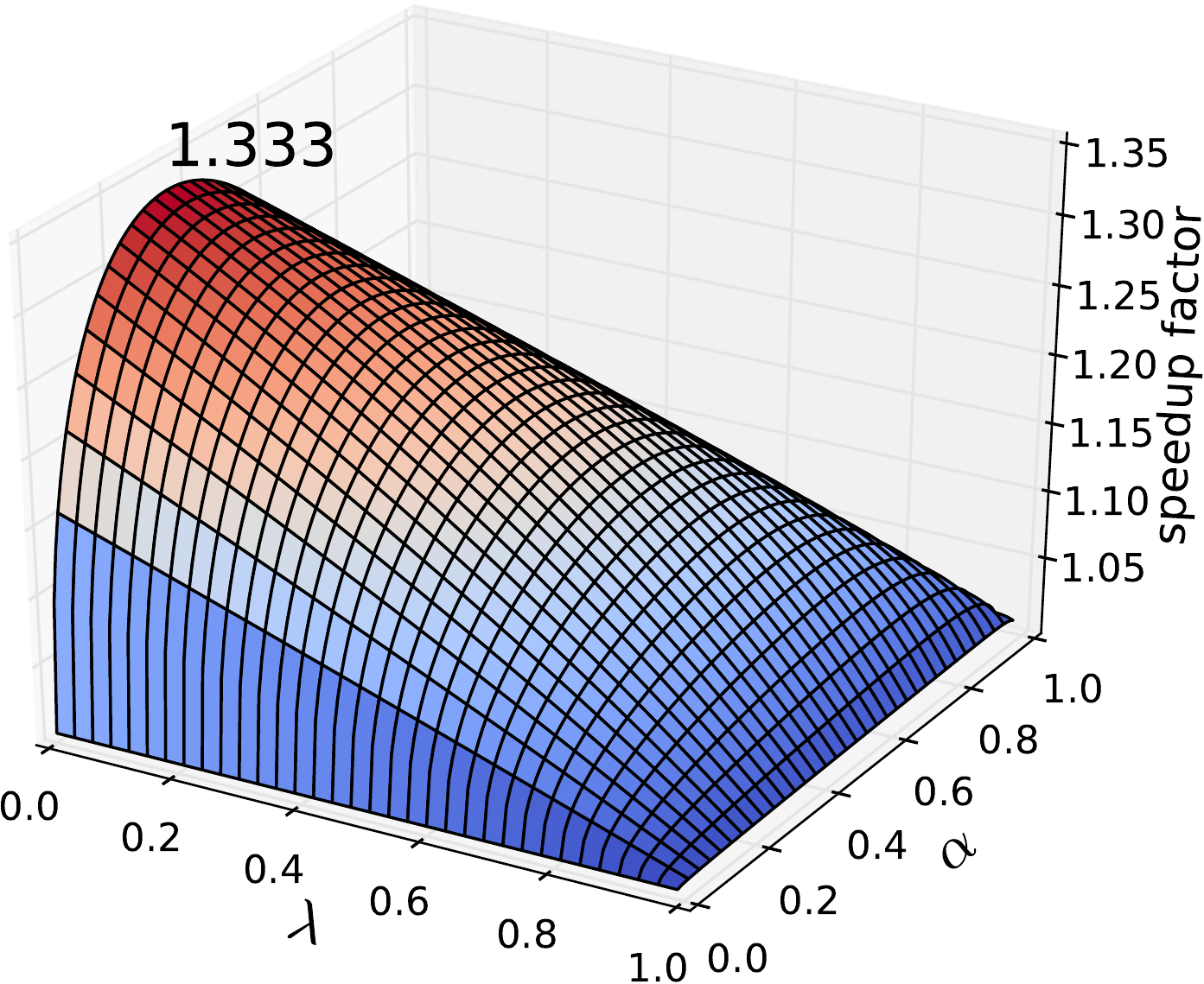}
\caption{3D image of speedup factor w.r.t $\alpha$ and $\lambda$}
\label{speedup_fig}
\end{figure}
\begin{table}
\centering
\scriptsize
\begin{tabular}[b]{|l|l|l|l|l|l|l|l|}
\hline
 \diagbox{$\lambda$}{$\alpha$} & 0.1 & 0.3 & $1/3$ & 0.5 & 0.7 & 0.9 & 1 \\ \hline
0 & 1.254 & 1.332 & \textbf{\color{red}1.333} & 1.309 & 1.227 & 1.091 & 1 \\ \hline
0.1 & 1.231 & 1.308 & 1.310 & 1.293 & 1.219 & 1.090 & 1 \\ \hline
0.3 & 1.183 & 1.256 & 1.259 & 1.254 & 1.201 & 1.087 & 1 \\ \hline
0.5 & 1.134 & 1.195 & 1.200 & 1.206 & 1.174 & 1.083 & 1 \\ \hline
0.7 & 1.082 & 1.126 & 1.130 & 1.143 & 1.133 & 1.074 & 1 \\ \hline
0.9 & 1.028 & 1.046 & 1.048 & 1.056 & 1.061 & 1.048 & 1 \\ \hline
1 & 1 & 1 & 1 & 1 & 1 & 1 & 1 \\ \hline
\end{tabular}
 \caption{Speedup factor w.r.t $\alpha$ and $\lambda$}
\label{speedup_table}
\end{table}
From Figure \ref{speedup_fig} and Table \ref{speedup_table}, we observe different trends for the speedup factor with respect to $\alpha$ and $\lambda$.
\begin{itemize}
 \item First, given a fixed $\lambda$, the speedup factor is not a monotonic function with respect to $\alpha$.
The relation between $\alpha$ and the speedup factor draws a downward parabola. Therefore, a straightforward conclusion regarding the impact of $\alpha$ on the speedup factor cannot
be drawn.
\item Given a fixed $\alpha$, the speedup factor is a monotonic decreasing function
with respect to increasing $\lambda$. It is seen that increasing $\lambda$ leads to
a smaller value of the speedup factor. \textit{This means that a larger $\lambda$ brings a positive effect on the schedulability of an IMC task set}.
\end{itemize}


\section{Extension to Elastic Mixed-Criticality Model}
\label{section:emc}
Su and Zhu in \cite{Su2013} introduced an \textit{Elastic Mixed-Criticality} (EMC) task model, where the \textit{elastic model} \cite{buttazzo2002elastic} is used to model low criticality tasks. When the MC system switches to high criticality mode,
low-criticality tasks scale up their original period to a larger
period such that low-criticality tasks continue to be scheduled with a degraded service (less frequently).  
Although the EMC model has been studied by \cite{Su2013}\cite{Su2014}\cite{jan2013maximizing}, there is not a utilization-based sufficient test
for the EMC model.
Therefore, in this section, we prove that the theories proposed in Section \ref{section:sufficientTest} apply to the EMC model \cite{Su2013} as well. 
Here, we use $T_i^{max}( \ge T_i)$ to denote the extended period of a \lo task $\tau_i$. Since, in the EMC model, the WCETs of a \lo tasks are the same in two modes, the utilization of \lo task $\tau_i$ in \hi mode is computed as $u_i^{HI}=C_i^{LO}/T_i^{max}$.
\begin{proposition}[Lemma 1 from \cite{Su2013}]
\label{proposition:emc}
 A set of EMC tasks is EMC schedulable under EDF-VD if $\uhh + \ulh \le 1$.
\end{proposition}
Here, in order to keep the consistence, we use $\ulh$ to denote $U(L,min)$ in \cite{Su2013}.
Proposition \ref{proposition:emc} is provided in \cite{Su2013} to check the schedulability of an EMC task set on a uniprocessor.
However, Proposition \ref{proposition:emc} is a necessary test. This means that even if
a given task set satisfies the condition presented in Proposition \ref{proposition:emc},
it is still possible that the task set is unschedulable under EDF-VD.
Below, we prove that the theories proposed in Section \ref{section:sufficientTest} can apply to the EMC model.

First, in \lo mode, since the EMC model just behaves like the classical MC model, Theorem \ref{theorem:lowcrit} holds for the EMC model.
Then we discuss the schedulability of the EMC model in \hi mode.
We have the following definition for the carry-over job of a low criticality task in the EMC model:
\begin{definition}
 In the EMC model, \textsf{carry-over job} $J_i$ of low criticality task $\tau_i$ has its release time $a_i < t_1$ and \textbf{original} deadline $d_i > t_1$.
\end{definition}
Then, we prove the following proposition for a \textsf{carry-over job}.
\begin{proposition}
\label{proposition:emc_carry}
 For an EMC \textsf{carry-over job} $J_i$, if it completes its execution before switch time instant $t_1$, then its \textbf{original} deadline $d_i$ is $\le (a_1 +x(t_2-a_1))$.
\end{proposition}
\begin{proof}
 Consider that \textsf{carry-over job} $J_i$ completes its execution before switch time instant $t_1$.
Suppose that $J_i$ has its \textbf{original} deadline $d_i>(a_1 + x(t_2-a_1))$.
Let $t^*$ denote the latest time instant at which $J_i$ starts to execute before $t_1$.
At time instant $t^*$, all jobs in $\mathcal{J}$ with deadlines $\le(a_1 + x(t_2-a_1))$ then have finished their execution.
Therefore, these jobs do not have any execution within interval $[t^*,t_2]$.
This implies that jobs in $\mathcal{J}$ which are released at or after $t^*$ can form a smaller job set and this smaller job set is sufficient to cause deadline miss at $t_2$.
This contradicts the minimality of $\mathcal{J}$.
Therefore, in this case we have
$d_i\le(a_1 + x(t_2-a_1))$
\end{proof}

\begin{lemma}
 \label{lemma:emc}
 Lemma \ref{lemma:low_crit} still holds for \lo tasks of the EMC model in \hi mode.
\end{lemma}
\begin{proof}
We can prove this lemma by doing some modifications on the proof of Lemma \ref{lemma:low_crit}.
Here, we mainly focus on the modified part. The proof uses the same notations explained in Section \ref{section:sufficientTest}.

For the EMC model, we need to consider a special case when \textsf{carry-over job} $J_i$ of \lo task $\tau_i $ has its extended deadline $d_i^{max} >t_2$.
Since $t_2$ is a deadline miss, a job with deadline $> t_2$ will not be scheduled within $[t_1,t_2)$ -see Proposition \ref{proposition:t1_t2}.
If $d_i^{max} >t_2$, job $J_i$ will not be executed after the
switch time instant $t_1$ and the maximum cumulative execution time of $\tau_i$ can be obtained as job $J_i$ completes its $C_i^{LO}$ before $t_1$.
Hence, the cumulative execution of task $\tau_i$ can
be bounded by,
\begin{equation}
\label{eq:emc_case2_1}
 \eta_i \le a_i\cdot u_i^{LO} + (d_i - a_i)u_i^{LO} = d_i \cdot u_i^{LO}
\end{equation}
By Proposition \ref{proposition:emc_carry}, we
replace $d_i$ with $(a_1+x(t_2-a_1))$ in Eq. (\ref{eq:emc_case2_1})
\begin{equation}
\small
\begin{split}
\label{eq:emc_second}
&\eta_i \le (a_1+x(t_2-a_1)) u_i^{LO} + (t_2 - (a_1+x(t_2-a_1)))u_i^{HI} \\
\Leftrightarrow & \eta_i \le (a_1 + x(t_2 - a_1))u_i^{LO} + (1 - x)(t_2 - a_1)u_i^{HI}
\end{split}
\end{equation}
The rest of the proof can follow the proof of Lemma \ref{lemma:low_crit}.
A complete proof can be found in Appendix II.

\end{proof}

Lemma \ref{lemma:emc} shows that Lemma \ref{lemma:low_crit} can still bound  the cumulative execution time of \lo tasks of the EMC model in \hi mode. Moreover, since there is no difference how the \hi tasks are scheduled in the EMC model or in the classical MC model, Proposition \ref{proposition:hi_crit} still holds for the \hi tasks in the EMC model. As a result, Theorem \ref{theorem:high_crit} holds for the EMC model as well. Then, we can directly obtain the following theorem,
\begin{theorem}
\label{theorem:emc}
 Theorem \ref{theorem:overall} is a sufficient test for the EMC model under EDF-VD.
\end{theorem}
Since Theorem \ref{theorem:overall} is a sufficient test for the EMC model under EDF-VD, the speedup factor results we obtained in Section \ref{section:speed} also 
apply to the EMC model, i.e., the speedup factor bound of the EMC model under EDF-VD is also $4/3$ by using our proposed sufficient test.

\section{Experimental Evaluation} 
\label{section:evaluation}
In this section, we conduct experiments to evaluate the effectiveness of the proposed sufficient test for the IMC model in terms of schedulable task sets (acceptance ratio). 
Moreover, we conduct experiments to verify the observations obtained in Section \ref{section:speed} regarding the impact of $\alpha$ and $\lambda$ on the average acceptance ratio. 
Our experiments are based on randomly generated MC tasks.
We use a task generation approach, similar to that used in \cite{Easwaran2013}\cite{Ekberg2013article}, to randomly generate IMC task sets to evaluate the proposed sufficient test. 
Each task $\tau_i$ is generated based on the following procedure,  
\begin{itemize}
 \item pCriticality is the probability that the generated task is a \hi task; pCriticality$\in[0,1]$.
 \item Period $T_i$ is randomly selected from range $[100,1000]$.
 \item In order to have sufficient number of tasks in a task set, utilization $u_i$ is randomly drawn from the range$[0.05,0.2]$.
 \item For any task $\tau_i$, $C_i^{LO} = u_i * T_i$.
 \item $R \ge 1$  denotes the ratio $C_i^{HI}/C_i^{LO}$ for every \hi task. 
 If $L_i=HI$, we set $C_i^{HI} = R * C_i^{LO}$. It is easy to see that $\alpha$ used in the speedup factor function is equal to $\frac{1}{R}$; 
 \item $\lambda \in (0,1]$ denotes the ratio $C_i^{HI}/C_i^{LO}$ for every \lo task. 
  If $L_i=LO$, we set $C_i^{HI} = \lambda * C_i^{LO}$.
\end{itemize}
In the experiment, we generate IMC task sets with different target utilization $U$. 
Each task set is generated as follows.
 Given a target utilization $U$, we first initialize an empty task set. Then, we generate task $\tau_i$ according to the task generation procedure introduced above and 
add the generated task to the task set. The task set generation stops as we have 
\[ U-0.05 \le U_{avg} \le U+0.05\] 
where \[U_{avg}= \frac{U^{LO}+U^{HI}}{2}\] is the average total utilization of the generated task set.  
If adding a new task makes $U_{avg} >U+0.05$, then the added task will be removed and a new 
task will be generated and added to the task set till the condition is met.  

\subsection{Comparison with AMC \cite{burns2013towards}}
To date, the modified AMC given in \cite{burns2013towards} is the only related work considering the schedulability of the IMC model under fixed-priority scheduling.
Therefore, in the first experiment, we compare EDF-VD by using our proposed test to the AMC approach in \cite{burns2013towards} in terms of average acceptance ratio. 
In this experiment, $R$ is randomly selected from a uniform distribution $[1.5, 2.5]$. 
With different $\lambda$ and pCriticality settings, we vary   
$U_{avg}$ from 0.4 to 0.95 with step of 0.05, to evaluate the effectiveness of the proposed sufficient test in terms of the average acceptance ratios. 
We generate 10,000 task sets for each given $U_{avg}$. 
Due to space limitations, we only present the experimental results when pCriticality$=0.5$. Results 
with different pCriticality settings can be found in Appendix III.  
The results are shown in Figure \ref{figure:utilization_05}, where the x-axis denotes the varying $U_{avg}$ and the y-axis denotes the acceptance ratio. 
In the figures, let EDF-VD and AMC denote our proposed schedulability test and the one proposed in \cite{burns2013towards}, respectively. 
In most cases, EDF-VD outperforms AMC in terms of acceptance ratio. We observe the following trends:
\begin{enumerate}
 \item When $U_{avg}\in[0.5,0.8]$, EDF-VD always outperforms AMC in terms of acceptance ratio. 
However, if $U_{avg} > 0.8$ and $\lambda=0.3$ or $0.5$, AMC performs better than EDF-VD. 
The same trend is also found for the classical MC model under EDF-VD and AMC, see the comparison in \cite{Ekberg2013article}.
\item By comparing sub-figures in Figure \ref{figure:utilization_05}, we see that the average acceptance ratio 
improves when $\lambda$ increases. This confirms the observation for the speedup factor we obtained in Section \ref{section:speed}. The increasing $\lambda$ leads to a smaller speedup factor. As a result, it provides a better schedulability. We need to notice that when $\lambda$ increases, not only EDF-VD improves its acceptance ratio but the acceptance ratio of AMC \cite{burns2013towards} also improves. 
\end{enumerate}


%
\begin{figure*}[t]
\centering
\begin{subfigure}[b]{0.62\columnwidth}
 \centering
 \includegraphics[width=\columnwidth]{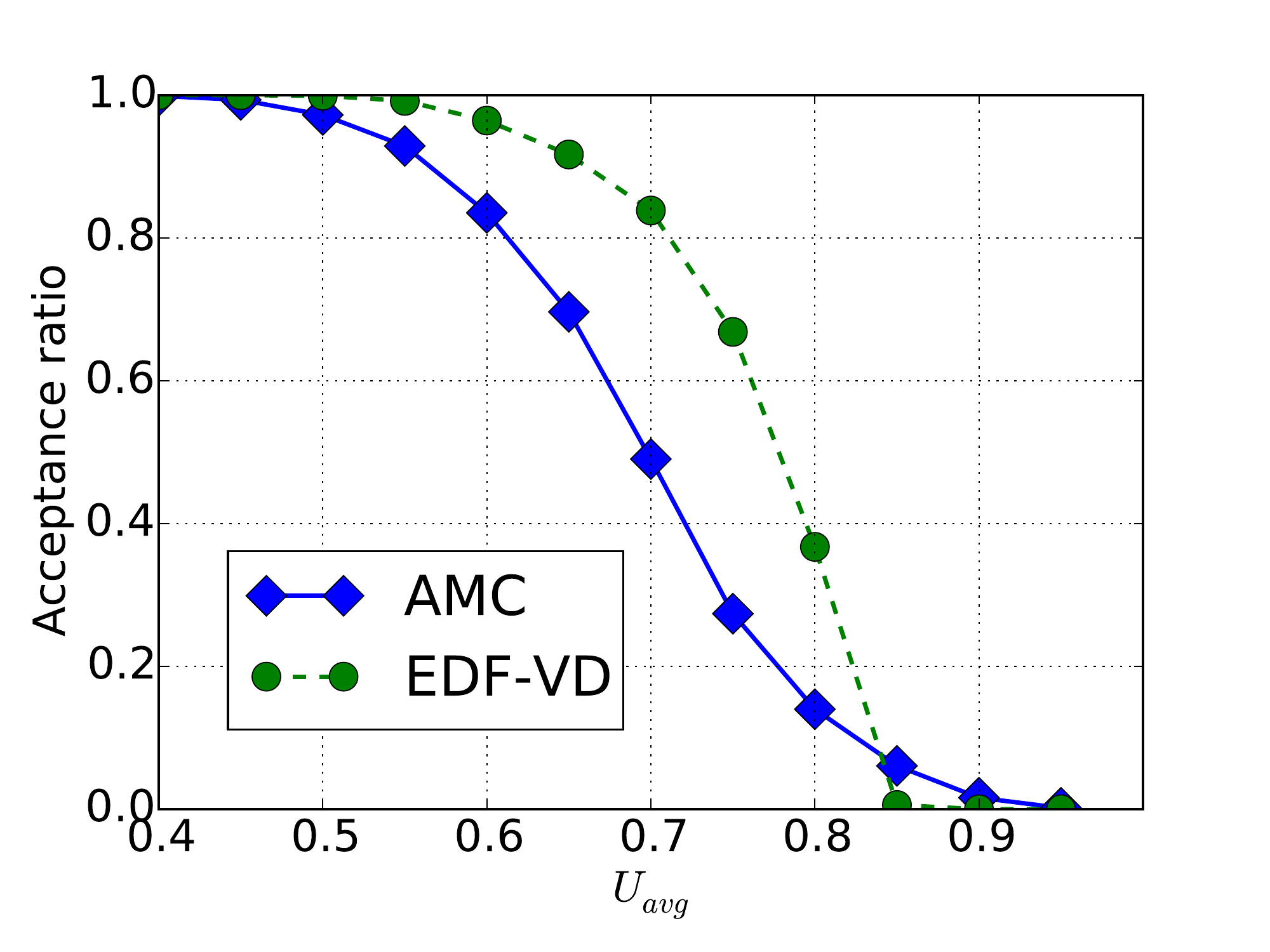}
 \caption{$\lambda = 0.3$}
 \label{05amc03}
\end{subfigure}
\begin{subfigure}[b]{0.62\columnwidth}
 \centering
\includegraphics[width=\columnwidth]{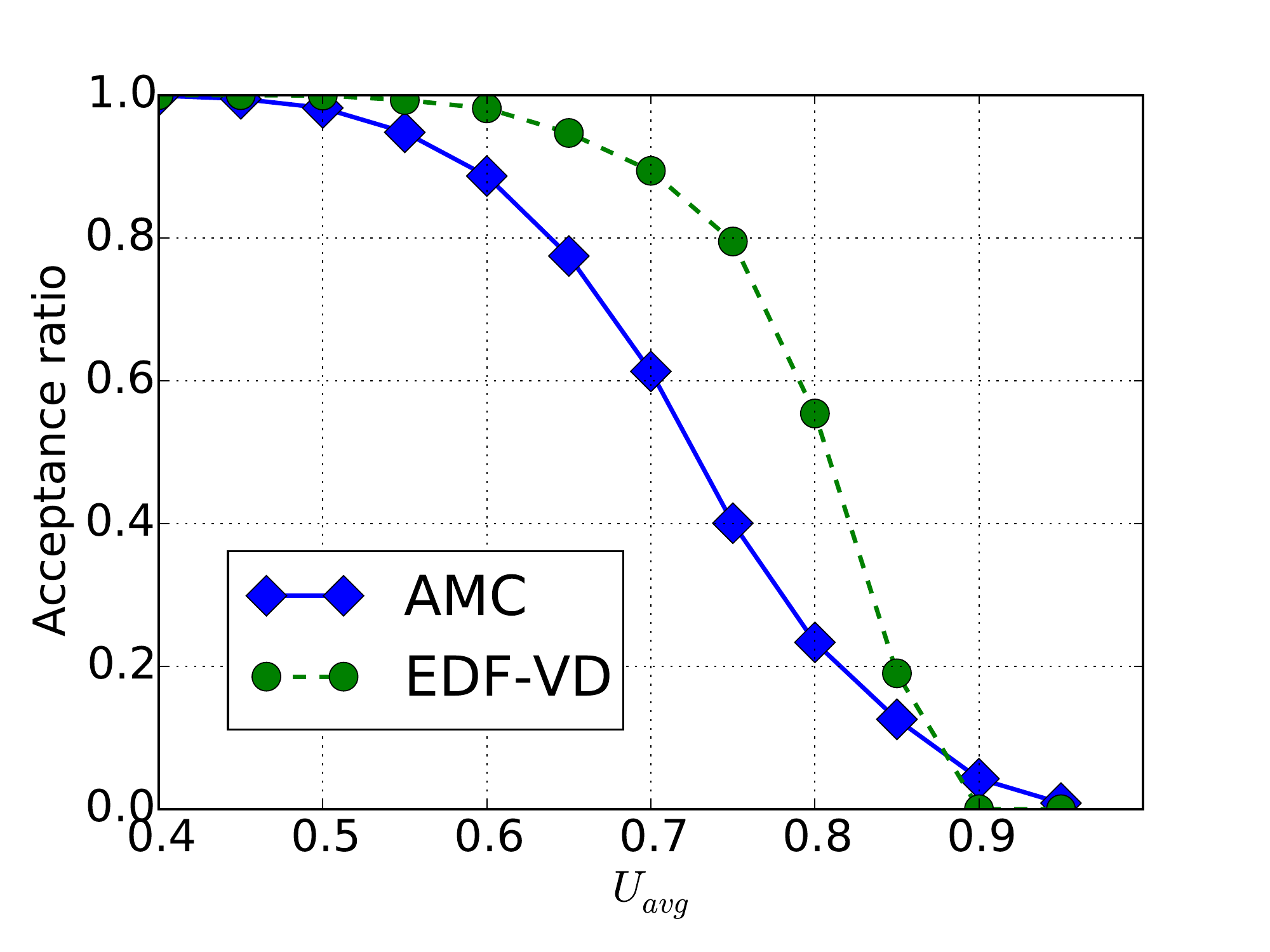}
\caption{$\lambda = 0.5$}
 \label{05amc05}
\end{subfigure}
\begin{subfigure}[b]{0.62\columnwidth}
 \centering
\includegraphics[width=\columnwidth]{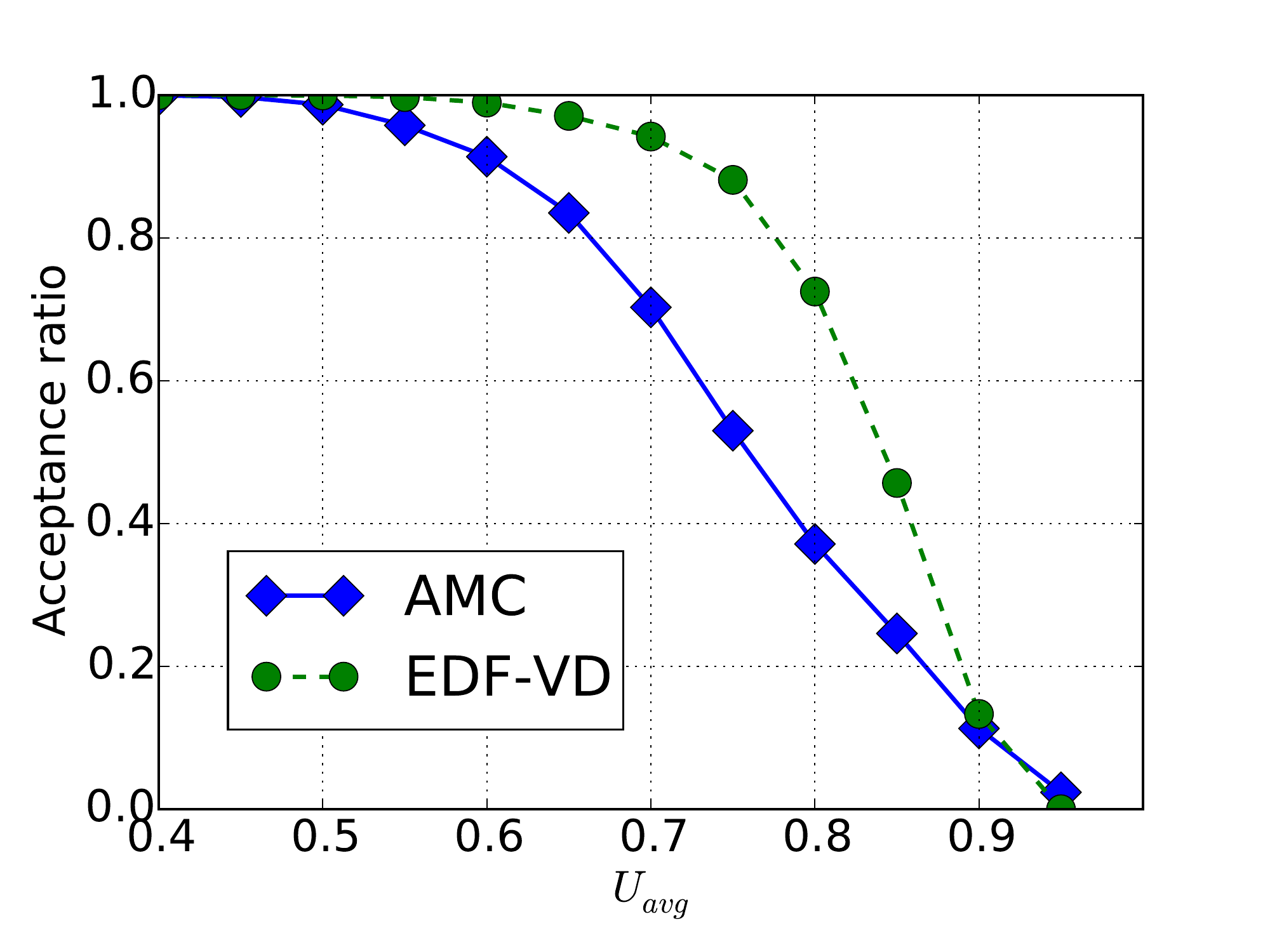}
\caption{$\lambda = 0.7$}
 \label{05amc05}
\end{subfigure}
\caption{Varying $U_{avg}$ with different $\lambda$ and pcriticality=0.5}
 \label{figure:utilization_05}
\end{figure*}

\subsection{Impact of $\alpha$ and $\lambda$}
Above, we compare our proposed sufficient test to the existing AMC approach. 
In this section, 
we conduct experiments to further evaluate the impact of $\lambda$ and $\alpha$ ($1/R$) on the 
acceptance ratio. 
In this experiment, we select $U_{avg}=\{ 0.65,0.7,0.75,0.8, 0.85\}$ to conduct experiments. 
We fix $U_{avg}$ to a certain utilization and vary $\lambda$ and $\alpha$ to evaluate the impact. 

We first show the results for $\lambda$.
The results are depicted in Figure \ref{lambda_pcrit05}, where the x-axis denotes 
the value of $\lambda$ from 0.2 to 0.9 with step of 0.1 and the y-axis denotes the average acceptance ratio. $R$ is randomly selected from a uniform distribution $[1.5, 2.5]$ and pCriticality$=0.5$. 
Similarly, 10,000 task sets are generated for each point in the figures. 
A clear trend can be observed that the acceptance ratio increases as $\lambda$ increases. This trend confirms the positive impact of increasing $\lambda$ on the schedulability which we have observed in Section \ref{section:speed}.


\begin{figure}
 \centering
\includegraphics[width=0.8\columnwidth]{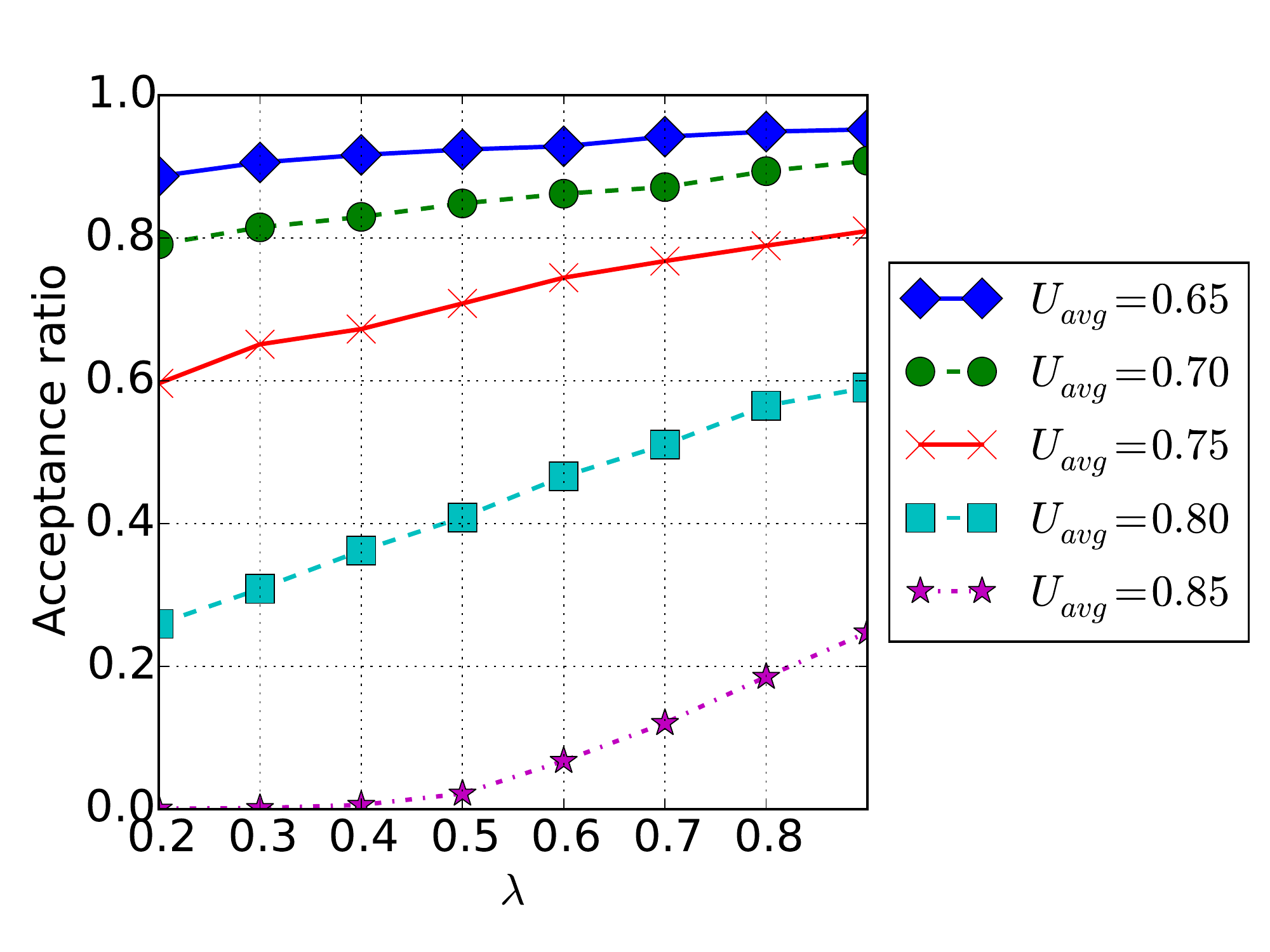}
\caption{Impact of $\lambda$}
 \label{lambda_pcrit05}
\end{figure}

Next we conduct experiments to evaluate the impact of $\alpha$ on the schedulability. 
Similarly, we fix $U_{avg}$ and vary $\alpha$ to carry out the experiments. 
Due to $\alpha = \frac{1}{R}$, if $\alpha$ is given, we compute the corresponding $R$ to generate task sets. 
The results are depicted in Figure \ref{fig:hi05}, where $\lambda=0.5$. 
The x-axis denotes the varying $\alpha$ from 0.1 to 0.9 with step of 0.1. while the y-axis denotes 
the average acceptance ratio.  
First, from Table \ref{speedup_table}, we see that with increasing $\alpha$ the speedup factor first increases till a point. This means within this range the scheduling performance of EDF-VD gradually decreases. After that point, 
the speedup factor decreases which means the scheduling performance of EDF-VD gradually improves. 
 The experimental results confirm what we have observed for $\alpha$ in Section \ref{section:speed}. 
The acceptance ratio gradually decreases till a point and then it increases.

\begin{figure}
 \centering
\includegraphics[width=0.8\columnwidth]{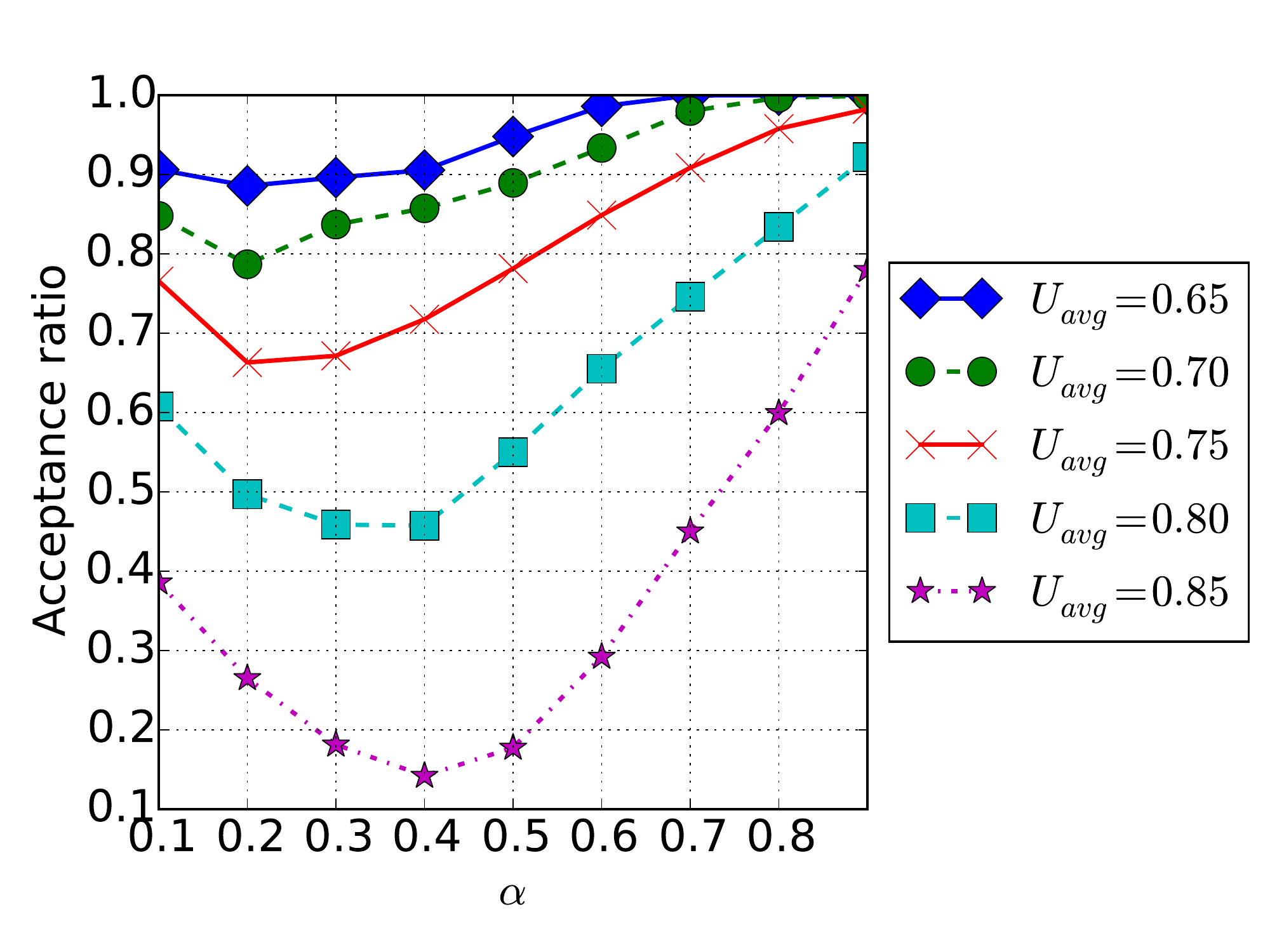}
\caption{Impact of $\alpha$}
 \label{fig:hi05}
\end{figure}

\section{Conclusions}
\label{section:conclusion}
In this paper, the \textit{imprecise mixed-criticality} (IMC) model from \cite{burns2013towards} is investigated. 
A sufficient test for the IMC model under EDF-VD is proposed and the proposed sufficient 
test later applies to the EMC model as well. 
Based on the proposed sufficient test, we derive a speedup factor function with respect to the utilization ratio $\alpha$ of all \hi tasks and the utilization ratio $\lambda$ of all \lo tasks. This speedup factor function provides a good insight to observe the impact of $\alpha$ and $\lambda$ on the speedup factor and quantifies suboptimality of EDF-VD for the IMC/EMC model in terms of speedup factor.
Our experimental results show that our proposed sufficient test outperforms the AMC approach in terms of acceptance ratio. 
Moreover, the extensive experiments also confirm the observations we obtained for the speedup factor. 

\bibliographystyle{IEEEtran}
\bibliography{IEEEabrv,main}
\newpage

\section*{Appendix I}
\begin{lemma}
\label{lemma:minimal_value}
 The minimum value of piece-wise function (\ref{eq:piecewise}) given in Section \ref{section:speed} is obtained when $b=b_0$.
\begin{equation}
\label{eq:piece}
s(b) =
\left\{
\begin{array}{ll}
\frac{(\alpha \lambda^2 - \alpha \lambda)b^2 + b - 1}{(\alpha \lambda - \alpha + 1)b - 1}  &  0< b \leq b_0 \\
\frac{(1-\alpha)b^2 + (\alpha \lambda +\alpha-1) b - \alpha}{(\alpha \lambda - \alpha + 1)b - 1}  & b_0<  b \leq 1
\end{array} \right.
\end{equation}
\end{lemma}
\begin{proof}
 For case of $0< b \leq b_0$, its derivative is
\[
s'(b) =  \frac{\alpha(\lambda-1)(\lambda(\alpha y - \alpha + 1)b^2 - 2\lambda b + 1)}{((\alpha \lambda - \alpha + 1)b - 1)^2}
\]
The denominator is obviously positive. For the numerator,
since the discriminant of $\lambda(\alpha \lambda - \alpha + 1)b^2 - 2 \lambda b + 1 = 0$
is $(2\lambda)^2 - 4\lambda(\alpha \lambda - \lambda + 1)$, which is negative since $0 < \lambda < 1$,
so we know $\lambda(\alpha \lambda - \alpha + 1)b^2 - 2 \lambda b + 1 > 0$. Moreover, we have $\lambda-1 < 0$, so putting them together we know the numerator is negative. In summary,
$s'(b)$ is negative and thus $s(b)$ is monotonically decreasing with respect to $b$ in the range $b \in (0,b_0]$.

For case of $b_0<  b \leq 1$,
we can compute the derivative of $s(b)$ by
\[
s'(b) =  \frac{(1-\lambda)((\lambda y - x + 1)b^2 -2b - (\lambda y -x -1))}{((\lambda y - x+1)b - 1)^2}
\]
The denominator is obviously positive.
For the numerator, we focus on $( x \lambda - x + 1)b^2 -2b - (x \lambda -x -1)$ part. The following equation
\[
(x \lambda - x + 1)b^2 -2b - (x \lambda -x -1) = 0
\]
has two roots $b_1 = 1$ and $b_2 = \frac{1+(x-x \lambda)}{1-(x-x\lambda)}$, which is greater than $1$, so we know $(x \lambda - x + 1)b^2 -2b - (x\lambda -x -1)$ is either always positive
or always negative in the range of
$b\in (b_0,1)$. Since we can construct
$(x\lambda - x + 1)b^2 -2b - (x\lambda -x -1) > 0$ with
$x=\lambda=b = 0.5$, so we know $(x \lambda - x + 1)b^2 -2b - (x \lambda -x -1)$ is always positive.
Moreover, since $1-x > 0$, the numerator
of $s'(b)$ is positive, so overall
$s'(b)$ is positive, and thus $s(b)$
is monotonically increasing with respect to $b$ in the range of $b\in (b_0,1]$.

In summary, we have proved $s(b)$ is monotonically
decreasing in $(0, b_0]$,
and monotonically increasing in $( b_0, 1]$, both
with respect to $b$, so the smallest value of $s(b)$
must occur at $b_0$.
\end{proof}

\begin{lemma}
\label{lemma:out_range}
If $0 <\alpha <1$ and $0 \le \lambda < 1$, then
 \begin{align}
       \small
	b_0^1 & = \frac{(2 - \alpha \lambda - \alpha) + (1-\lambda)\sqrt{-3\alpha^2 + 4\alpha }}{2(-\alpha \lambda^2 + \alpha \lambda - \alpha +1)} \label{e:one} >1\\
	b_0^2 & = \frac{(2 - \alpha \lambda - \alpha) - (1-\lambda)\sqrt{-3\alpha^2 + 4\alpha }}{2(-\alpha \lambda^2 + \alpha \lambda - \alpha +1)} \label{e:two} \in [0,1]
\end{align}
\end{lemma}
\begin{proof}
We start with proving $b_0^1 > 1$. We first prove $b_0^1 \geq 0$
by showing both the numerator and dominator are positive. 
For simplicity, we use $N_1$ and $M_1$ to 
denote the numerator and denominator of $b_0^1$ in (\ref{e:one}), 
and  $N_2$ and $M_2$ the numerator and denominator of $b_0^2$
in (\ref{e:two}). 
Note that the following reasoning relies on that $\alpha \in (0,1), \lambda\in[0,1)$. 
\begin{enumerate}
	\item $N_1 > 0$. First, we have 
	\[
	\begin{split}
	& N_1 \times N_2 \\
	= ~& (2 - \alpha \lambda - \alpha)^2 - (1-\lambda)^2(-3\alpha ^2 + 4\alpha  ) \\
	= ~& 4 \alpha  \lambda (1-\lambda)(1-\alpha ) +  4 (1-\alpha )^2 \\
	> ~& 0
	\end{split} \]
	Moreover, it is easy to see $N_2 > 0$. Therefore, we can conclude 
	that  $N_1$ is also positive.  
	  
	\item $M_1 > 0$.  
	$2(-\alpha \lambda^2 + \alpha \lambda - \alpha  +1) = 2(\alpha \lambda (1 - \lambda) + (1- \alpha )) $, which is positive. 
\end{enumerate} 
In summary, both the numerator and the denominator of $b_0^1$ in 
(\ref{e:one}) are positive, so $b_0^1 \geq 0$.
Next we prove  $b_0^1 \leq 1$ by showing $N_1 - M_1 \leq 0$: 
\[\begin{split} 
&N_1 - M_1 \\
= ~& (\lambda-1) (\sqrt{-3\alpha ^2 + 4\alpha  }  + \alpha (2\lambda - 1))  
\end{split} \]
which is negative if $\lambda \geq 0.5$ (since $\lambda-1 < 0$ and $\sqrt{-3\alpha ^2 + 4\alpha }  + \alpha (2\lambda - 1) \geq  0 $). So in the following we focus on the case of  $\lambda < 0.5$. Since $\lambda < 0.5$, we know $\alpha  (2\lambda - 1)$ is negative, so we define two positive number $A$ and $B$ as follows
	\begin{align}
	A  &= \sqrt{-3\alpha^2 + 4\alpha }  \\
	B  &= \alpha(1 - 2\lambda)
	\end{align}
	so $N_1 - M_1 = (\lambda-1)(A - B)$. Since $\lambda-1 < 0$, we only need to prove $A - B > 0 $, which is equivalent to proving $A^2 - B^2 > 0 $ (as both $A$ and $B$ are positive):
	$A^2 - B^2 > 0 $, which is done as follows:
	\[\begin{split}
	A^2 - B^2  = &-3\alpha^2 + 4\alpha  - \alpha^2(2\lambda - 1)^2 \\
	=  &4\alpha (1 - \alpha) + 4\alpha^2\lambda(1 - \lambda)\\
	>  & 0 
	\end{split}\]
	so we have $A - B > 0$ and thus $N_1 - M_1 = (\lambda-1)(A - B) < 0$. 
	In summary, we have proved $N_1 - M_1 < 0$ for the cases of both 
	$\lambda \geq 0.5$ and $\lambda < 0.5$, so we know $b_0^1 \in[0, 1]$.

%
	Next we prove $b_0^2 > 1$, by showing $N_2 - M_2 > 0$ 
	\[\begin{split}
	& N_2 - M_2 \\
	= ~& (1 - \lambda)( \sqrt{-3\alpha^2 + 4\alpha }  - \alpha(2\lambda-1))
	\end{split}\]
	If $\lambda \leq 0.5$, then $\sqrt{-3\alpha^2 + 4\alpha }  - \alpha(2\lambda-1) > 0$, and since $1-\lambda > 0$ we have $N_2 - M_2 > 0$.
	If $\lambda > 0.5$, we let $C = \alpha(2\lambda -1) > 0$ and also use $A$
	as defined above, $N_2 - M_2 = (1 - \lambda)(A - C)$. 
	To prove $A - C > 0$, it suffices to prove $A^2 - C^2 > 0$, as shown in the following:
	\[\begin{split}
	A^2 - C^2 = ~&  - 3\alpha^2 + 4\alpha -  \alpha^2(2\lambda -1)^2 \\
	= ~&  4\alpha - (3 + (2\lambda -1)^2)\alpha^2 \\
	> ~&  4\alpha - 4\alpha^2 ~~(\lambda < 1 \textrm{ ,so } 2\lambda-1 < 1) \\
	> ~&  0
	\end{split}\]
	By now we have proved $ N_2 - M_2$ for both cases of $\lambda \leq 0.5$ and $\lambda > 0.5$, so we known $b_0^2 > 1$.

\end{proof}

\begin{figure*}[t]
\centering
\begin{subfigure}[b]{0.62\columnwidth}
 \centering
 \includegraphics[width=\columnwidth]{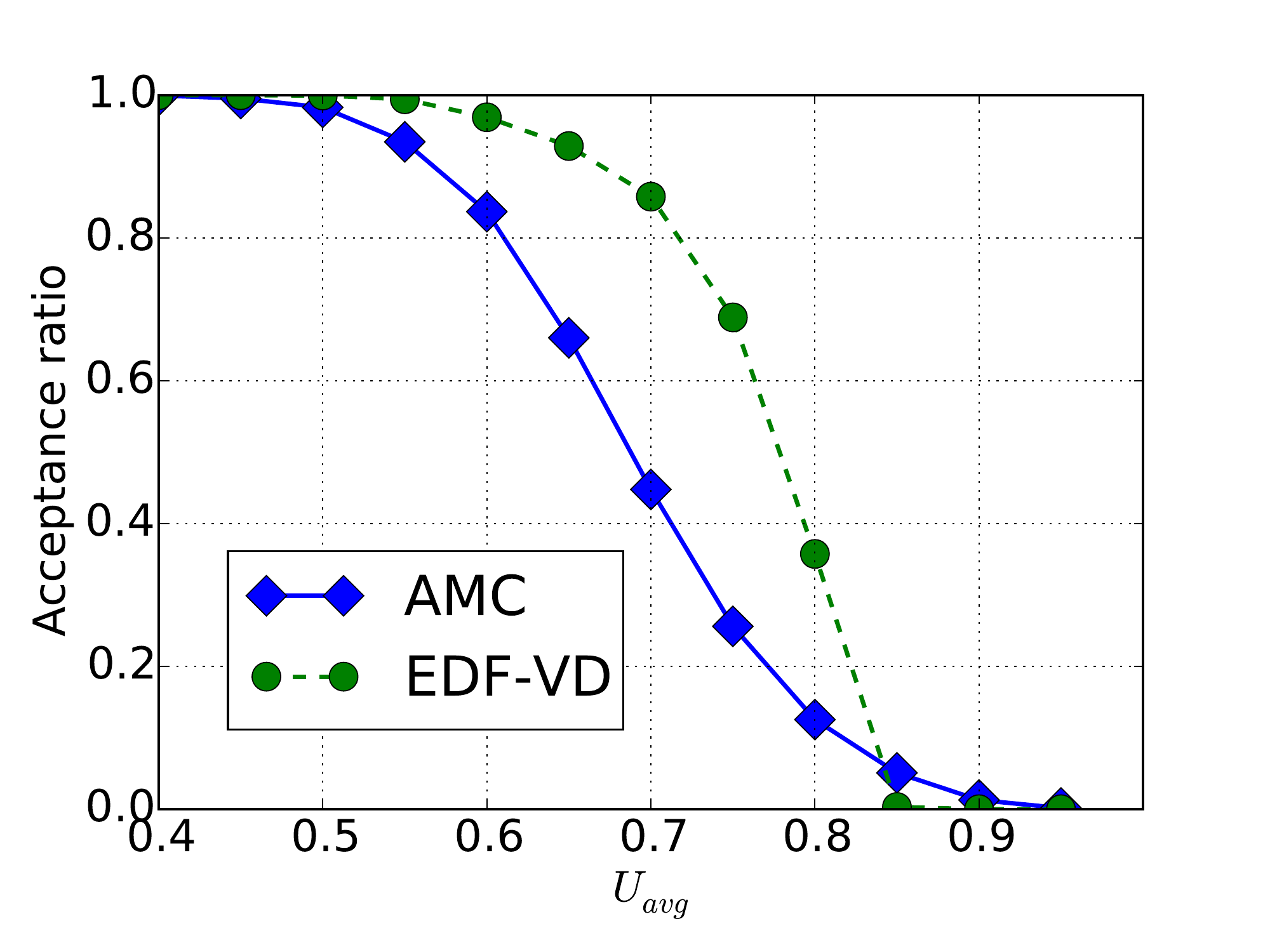}
 \caption{$\lambda = 0.3$}
 \label{03amc03}
\end{subfigure}
\begin{subfigure}[b]{0.62\columnwidth}
 \centering
\includegraphics[width=\columnwidth]{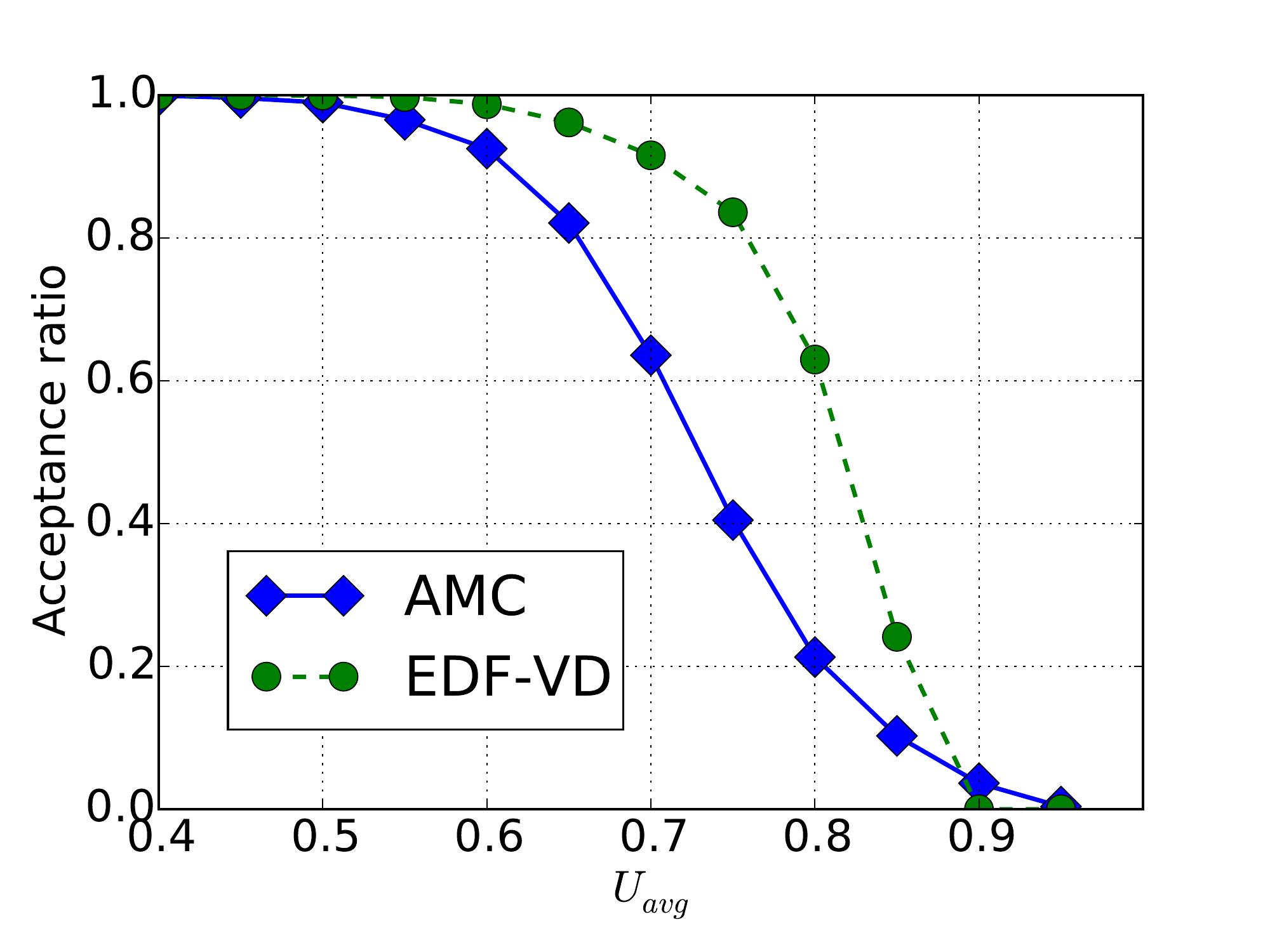}
\caption{$\lambda = 0.5$}
 \label{03amc05}
\end{subfigure}
\begin{subfigure}[b]{0.62\columnwidth}
 \centering
\includegraphics[width=\columnwidth]{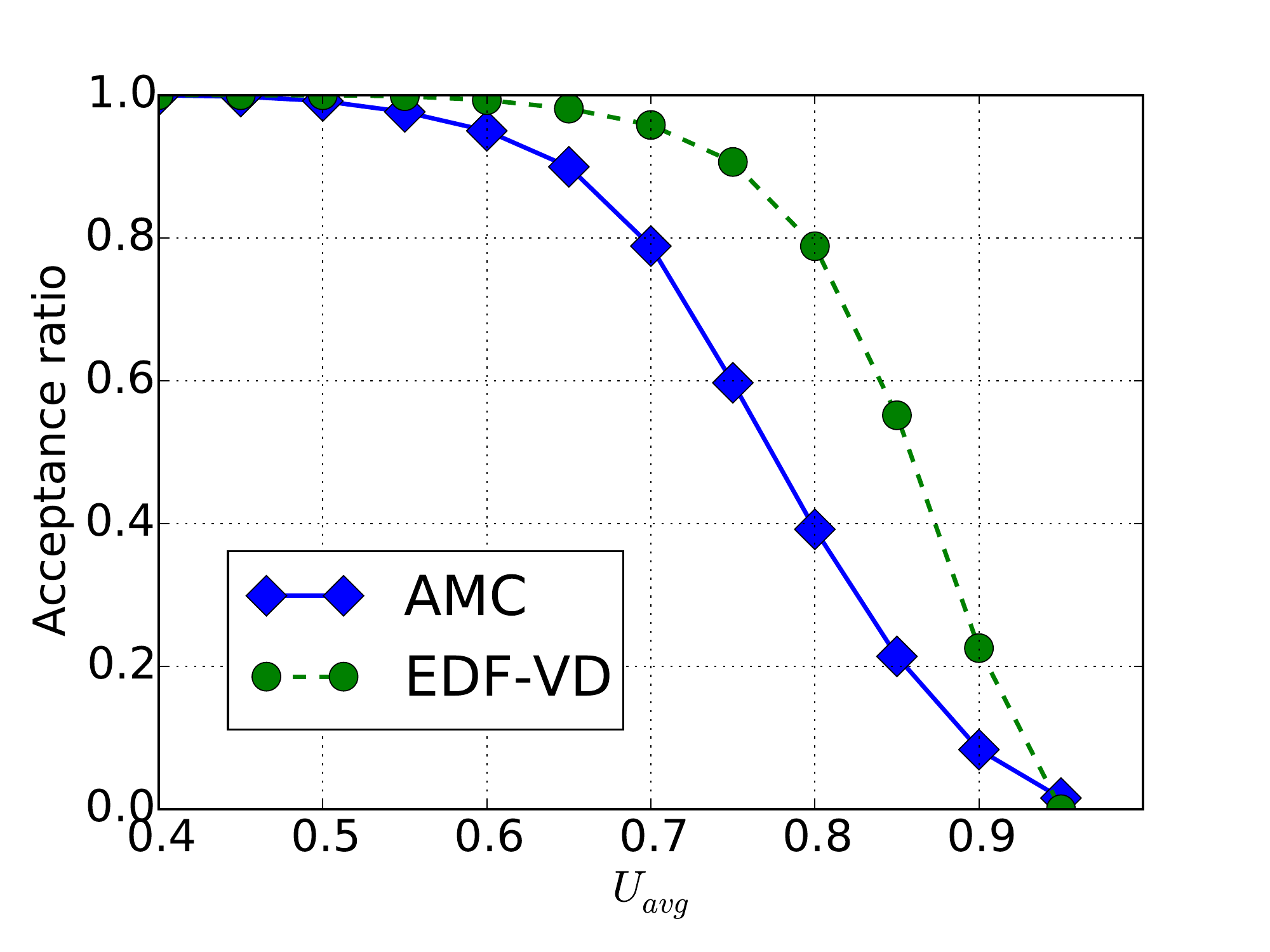}
\caption{$\lambda = 0.7$}
 \label{03amc05}
\end{subfigure}
\caption{Varying $U_B$ with different $\lambda$ and pcriticality=0.3}
 \label{figure:utilization_03}
\end{figure*}

\section*{Appendix II}
The following is the complete proof of Lemma \ref{lemma:emc}.
\begin{proof}
 We use the same notations explained in Section \ref{section:sufficientTest}.
When $u_i^{LO}=u_i^{HI}$, it is trivial to see that Lemma \ref{lemma:low_crit} holds for the EMC model. Now we focus on the case when $u_i^{LO} > u_i^{HI}$
To prove this case, we need to consider two cases where \lo task $\tau_i$ releases a job within interval $(a_1,t_2]$ or it does not.
\begin{itemize}
\item \textbf{Case 1} ($\tau_i$ releases a job within interval $(a_1,t_2]$): If there is no \textsf{carry-over job} , the proof is the same as we have proved for the IMC model (see the proof of Sub-case 1 in Lemma \ref{lemma:low_crit}). Here, we focus on the case that there is a \textsf{carry-over job}. Let $J_i$ denote the \textsf{carry-over job} with absolute release time $a_i$, original deadline $d_i$, and
maximum deadline $d_i^{max}$.
Here, we consider two cases, $d_i^{max} >t_2$ and $d_i^{max} \le t_2$.
\begin{itemize}
 \item $d_i^{max} >t_2$: since $t_2$ is a deadline miss, a job with deadline $> t_2$ will not be scheduled within $[t_1,t_2)$ -see Proposition \ref{proposition:t1_t2}.
If $d_i^{max} >t_2$, job $J_i$ will not be executed after the
switch time instant $t_1$ and the maximum cumulative execution time of $\tau_i$ can be obtained as job $J_i$ completes its $C_i^{LO}$ before $t_1$.
Hence, the cumulative execution of task $\tau_i$ can
be bounded by,
\begin{equation}
\label{eq:emc_case2_1}
 \eta_i \le a_i\cdot u_i^{LO} + (d_i - a_i)u_i^{LO} = d_i \cdot u_i^{LO}
\end{equation}
By Proposition \ref{proposition:emc_carry}, we
replace $d_i$ with $(a_1+x(t_2-a_1))$ in Eq. (\ref{eq:emc_case2_1})
\begin{equation}
\small
\begin{split}
\label{eq:emc_second}
&\eta_i \le (a_1+x(t_2-a_1)) u_i^{LO} + (t_2 - (a_1+x(t_2-a_1)))u_i^{HI} \\
\Leftrightarrow & \eta_i \le (a_1 + x(t_2 - a_1))u_i^{LO} + (1 - x)(t_2 - a_1)u_i^{HI}
\end{split}
\end{equation}

\item $d_i^{max} \le t_2$: in this case, the cumulative execution of \lo task $\tau_i$ can be bounded as follows:
\begin{equation}
\small
 \label{eq:emc_final}
  \begin{split}
  & \eta_i \le a_i u_i^{LO} + (t_2 - a_i)u_i^{HI} \\
  &( \text{Since }a_i < t_1 \\
  &\text{ and } t_1 < (a_1+x(t_2-a_1)) \text{ from Proposition \ref{proposition:case_1_2}}) \\
  \Rightarrow& \eta_i \le \big(a_1 + x(t_2 - a_1)\big)u_i^{LO} + \big(t_2 - \big(a_1 + x(t_2 - a_1)\big)\big)u_i^{HI} \\
  \Leftrightarrow & \eta_i \le (a_1 + x(t_2 - a_1))u_i^{LO} + (1 - x)(t_2 - a_1)u_i^{HI}
  \end{split}
\end{equation}
\end{itemize}
 \item \textbf{Case 2} ($\tau_i$ does not release a job within interval $(a_1,t_2]$):
For \lo task $\tau_i$, let $J_i$ denote the last release job before $a_1$, where $a_i(<a_1)$ and $d_i$ are the absolute release time and deadline of $J_i$, respectively. Moreover, let $d_i^{max}(>d_i)$ denote the new absolute deadline of job $J_i$ as the system switches to \hi mode. Here, there are two cases, $d_i \le t_1$ and $d_i > t_1$. For $d_i\le t_1$, the cumulative execution of task $\tau_i$ can be computed as follows:
\begin{equation}
\label{eq:emc_case0}
 \eta_i = d_i\cdot u_i^{LO}
\end{equation}
For $d_i > t_1$, if $d_i^{max}\le t_2$, then the maximum cumulative execution can be bounded as follows:
\[\small
\begin{split}
   &\eta_i \le a_i\cdot u_i^{LO} + (d_i^{max} - a_i)u_i^{HI}\\
  \Rightarrow &\eta_i \le a_i\cdot u_i^{LO} + (t_2 - a_i)u_i^{HI} \quad(\text{since }d_i^{max}\le t_2)
  \end{split}\]
Since $a_i < t_1 \le (a_1 +x(t_2 -a_1))$ by Proposition \ref{proposition:case_1_2}, we obtain
\begin{equation}
\label{eq:emc_first}
\small
\begin{split}
& \eta_i \le (a_1 +x(t_2 -a_1))u_i^{LO} + (t_2 - (a_1+x(t_2-a_1)))u_i^{HI} \\
\Leftrightarrow & \eta_i \le (a_1 + x(t_2 - a_1))u_i^{LO} + (1 - x)(t_2 - a_1)u_i^{HI}
\end{split}
\end{equation}
If $d_i^{max}> t_2$, our reasoning is similar to the case discussed in Case 1.
The maximum cumulative execution happens to that $J_i$ completes its execution before $t_1$.
Similarly, in this case, its cumulative execution can be upper bounded by
\[
\eta_i \le a_i u_i^{LO} + (d_i-a_i)u_i^{LO}\]
By Proposition \ref{proposition:emc_carry}, we obtain
\[\small
\begin{split}
&\eta_i \le (a_1+x(t_2-a_1)) u_i^{LO} + (t_2 - (a_1+x(t_2-a_1)))u_i^{HI} \\
\Leftrightarrow & \eta_i \le (a_1 + x(t_2 - a_1))u_i^{LO} + (1 - x)(t_2 - a_1)u_i^{HI}
\end{split}
\]
\end{itemize}
The above discussion proves that
Lemma 1 still holds for \lo tasks of the EMC model.
\end{proof}

\section*{Appendix III}
Experimental results between EDF-VD and AMC is depicted in Figure \ref{figure:utilization_03}, where pCriticality$=0.3$.

\end{document}